\documentclass[11pt]{article}

\usepackage[utf8]{inputenc}
\usepackage[margin=1in]{geometry}
\usepackage{amsmath}
\usepackage{amsfonts}
\usepackage{amsthm}
\usepackage{amssymb}
\usepackage{mathtools}
\usepackage{cases}
\usepackage{subcaption}
\usepackage{relsize}
\usepackage{titling}
\usepackage{empheq}
\usepackage{slashed}
\usepackage{physics}
\usepackage{authblk}
\usepackage[style=ext-numeric, giveninits=true, doi=false,isbn=false,url=false,eprint=false,articlein=false]{biblatex}
\DeclareFieldFormat[article, inbook]{title}{#1} 
\usepackage[capitalise]{cleveref}
\usepackage[colorinlistoftodos,prependcaption,textsize=tiny]{todonotes}
\usepackage[toc,page,header]{appendix}
\makeatletter
\newtheorem*{rep@theorem}{\rep@title}
\newcommand{\newreptheorem}[2]{%
\newenvironment{rep#1}[1]{%
 \def\rep@title{#2 \ref{##1}}%
 \begin{rep@theorem}}%
 {\end{rep@theorem}}}
\makeatother
\newtheorem{lemma}{Lemma}[section]
\newtheorem{thm}{Theorem}[section]
\newtheorem{prop}{Proposition}[section]

\theoremstyle{definition}
\newtheorem{defn}{Definition}[section]
\newtheorem*{remark*}{Remark}
\newreptheorem{thm}{Theorem}
\crefname{lemma}{Lemma}{Lemmas}
\crefname{thm}{Theorem}{Theorems}
\crefname{prop}{Proposition}{Propositions}
\newcommand{\subtitle}[1]{%
  \posttitle{%
    \par\end{center}
    \begin{center}\large#1\end{center}
    \vskip0.5em}%
}
\newcommand{\s}{\tilde{s}}
\newcommand{\comment}[1]{}
\DeclareMathOperator{\artanh}{artanh}
\DeclareMathOperator{\img}{im}

\DeclareMathOperator{\supp}{supp}
\DeclareMathOperator{\realpart}{Re}
\DeclareMathOperator{\imag}{Im}
\numberwithin{equation}{section}
\title{Existence of Zero-damped Quasinormal Frequencies for Nearly Extremal Black Holes}
\author{Jason Joykutty\thanks{Department of Applied Mathematics and Theoretical Physics, University of Cambridge, Wilberforce Road, Cambridge CB3 0WA, United Kingdom. Email: \texttt{jj441@cam.ac.uk}.}}
\date{December 2021}
\begin{document}
\maketitle
\begin{abstract}
    \noindent It has been observed that many spacetimes which feature a near-extremal horizon exhibit the phenomenon of zero-damped modes. This is characterised by the existence of a sequence of quasinormal frequencies which all converge to some purely imaginary number $i\alpha$ in the extremal limit and cluster in a neighbourhood of the line $\imag s=\alpha$. In this paper, we establish that this property is present for the conformal Klein-Gordon equation on a Reissner-Nordstr\"om-de Sitter background. This follows from a similar result that we prove for a class of spherically symmetric black hole spacetimes with a cosmological horizon. We also show that the phenomenon of zero-damped modes is stable to perturbations that arise through adding a potential.
\end{abstract}
\section{Introduction}
Like many physical systems, black holes emit radiation as a response to perturbations. These gravitational waves are expected to be dominated by (after a short initial burst of radiation) quasinormal modes, which are analogous to the characteristic normal modes of oscillation of idealised musical instruments. Just as a normal mode has a (purely imaginary) frequency associated with it, a quasinormal mode has a complex frequency $s$ which captures the dissipative nature of the black hole horizon\footnote{Here we use the notation of \cite{cmw} since we are considering an initial value problem, where the Laplace transform is more convenient. The Fourier transform is more commonly used in the literature with frequency $\omega=is$.} - the real part indicates exponential decay (or growth) of the oscillations, since the time dependence of these modes is of the form $e^{st}$.\\\\
Quasinormal modes were first observed by Vishveshwara \cite{Vishveshwara} from numerical simulations of gravitational waves scattering due to a Schwarzschild black hole. In this case, the scattered waveform at intermediate to late times consisted of a damped sinusoid with frequency unrelated to the parameters of the incident wavepacket. This is a phenomenon typical in black hole dynamics: detected gravitational wave signals have been dominated by a superposition of decaying sinusoids in what is called the ringdown phase \cite{PhysRevLett.116.221101, PhysRevLett.125.101102}. The expected behaviour of the signal at late times is dependent on the asymptotics of the background spacetime (see \cite{Ching_1995}) - for example, we expect this ringdown behaviour to continue if the spacetime has an asymptotically de Sitter end (for rigorous results, see for example \cite{hintz2021kds,mavrogiannis,petersenvasykds}), however we expect inverse polynomial decay in time if it is asymptotically flat (this is Price's law - for examples of rigorous results on this, see \cite{aagkerr, aagrn, Hintzprice, yakovritakerr}).\\\\
The spectrum of quasinormal frequencies of an astrophysical black hole (one whose charge is small compared to its mass) is distinct from that of other sources of gravitational radiation and is characterised completely by its parameters: mass and angular momentum. Thus the spectrum gives us predictions to compare with gravitational wave observations, serving as further tests for the existence of black holes and general relativity \cite{PhysRevD.103.122002, PhysRevD.98.084038, spec}. Given a radiation source we are confident is a black hole, we can also use these predictions as a test of the no-hair theorems \cite{spec, nohair}. For recent work which analyses observational data from LIGO-Virgo and compares this to prediction, see \cite{PhysRevD.103.122002, PhysRevD.98.084038, PhysRevD.103.124041}.\\\\
\comment{Whilst a black hole perturbation can be modelled by a metric perturbation obeying the linearised Einstein equations, it suffices to consider the wave equation with a potential for spherically symmetric background spacetimes. In summary, the perturbation $h_{ab}$ can be decomposed into scalar, vector and tensor spherical harmonics according to how different combinations of the components transform under rotations: $h_{tt}, h_{tr}, h_{rr}$ transform like scalars, $(h_{t\theta},h_{t\phi}),(h_{r\theta},h_{r\phi})$ like vectors and the remaining four components not determined by symmetry like a $2\times 2$ tensor where $t$ is the usual static time, $r$ the radial coordinate and $(\theta,\phi)$ coordinates on $S^2$. These can be further split into two parities depending on how they transform under a point reflection: axial and polar perturbations. By choosing a suitable gauge and noticing redundancies in the equations, the information from the axial perturbations can be encoded in the Regge-Wheeler equation - a one-dimensional wave equation first studied in \cite{ReggeWheeler}. Similarly (though involving more complicated calculations) the polar perturbations can be reduced to a one-dimensional wave equation with a different potential called the Zerilli equation, first studied in \cite{zerilli}. For a more detailed discussion of these cases, consult \cite{ReggeWheeler, zerilli} or the review \cite{Nollert_1999}. \\\\}
Zero-damped modes are a family of quasinormal modes characterised by frequencies which cluster on a line of constant imaginary part and converge to a pure imaginary number as a horizon approaches extremality. They have been observed when computing the spectrum of quasinormal frequencies in a variety of situations: for example with various fields in Reissner-Nordstr\"om \cite{hod6, hod5, hod7, pzimmerman}, in Kerr \cite{hod3, hod4, hod1, azimmerman} and in Kerr-Newman \cite{santos2015, hod2}. The existence of these frequencies has implications on the decay of perturbations for spacetimes with asymptotically de Sitter ends and thus is tied to the strong cosmic censorship conjecture for solutions to the Einstein equations with cosmological horizons. This relationship was explored for Reissner-Nordstr\"om-de Sitter black holes in \cite{destounis1, destounis2, destounis3, destounis5}. These examples prompt the question of what useful statements can be made about the generic behaviour of quasinormal frequencies as a horizon approaches extremality (in the sense that the surface gravity $\kappa\rightarrow 0$). In particular, one can ask if the phenomenon of zero damped modes is generic.\\\\
In this paper, we initially consider a conformally-coupled Klein-Gordon field propagating in a de Sitter background of cosmological constant $3\kappa^2$ with a potential $V$:
\begin{align}
    -\Box_g\psi+2\kappa^2\psi+V\psi=0.\label{CKG}
\end{align}
This is a simplified model of the linearised Einstein equations which avoids the tensorial structure, but the constructions we shall consider are generalisable to other test fields of interest. Furthermore, the full tensorial equations for the linearised Einstein equations can be reduced to a set of wave equations with different potentials depending on the type of perturbation (see \cite{ReggeWheeler, zerilli} or the review \cite{Nollert_1999} for the Schwarzschild case). It can be seen with relatively simple arguments (see \cref{reg}) that when $V=0$, equation \eqref{CKG}  exhibits zero-damped quasinormal frequencies $s\in -\kappa \mathbb{N}$.\\\\
The first result of this paper (\cref{mainresultthm}) is that \eqref{CKG} exhibits the phenomenon of zero-damped frequencies for $V\in C^{\infty}(\mathbb{R}^3;\mathbb{R})$ decaying sufficiently rapidly i.e. this phenomenon is stable to suitable perturbations. The condition on the potential is essentially that $V$ and its derivatives decay faster than an inverse square potential, since these potentials can be treated as ``small" in the extremal limit. More precisely we require
\begin{align*}
    |\mathbf{x}|^{|\alpha|+2}\partial^{\alpha}V(\mathbf{x})\rightarrow 0 \quad \text{as} \quad |\mathbf{x}|\rightarrow 0
\end{align*}
for all multi-indices $\alpha$. We also find that if we impose slightly better decay (at least as fast as $O(1/|\mathbf{x}|^3)$) the size of the perturbation of these frequencies from the unperturbed ones are of $O(\kappa^2)$ where $\kappa$ is the surface gravity of the horizon. For $V$ spherically symmetric, we can also obtain an expansion in powers of $\kappa$ for the perturbed frequencies (see \cref{seriesexpansion}).\\\\
The main result (\cref{mainresultthmblackhole}) is to extend the above methods to a generic spherically symmetric black hole spacetime with an asymptotically de Sitter end. We again consider the conformal Klein-Gordon equation but this time on a spacetime with a metric of the form
\begin{align}
    g=-f(r)dt^2+\frac{dr^2}{f(r)}+r^2g_{S^2}
\end{align}
where $f(r)=1+w_{\Lambda}(r)+\Lambda\alpha_{\Lambda}r/3-\Lambda r^2/3$ and $\Lambda$ is the cosmological constant of the spacetime. The equation we analyse is
\begin{align}
    -\Box_g\psi+\frac{R}{6}\psi=0\label{ConformalKleinGordon}
\end{align}
where $R$ is the Ricci scalar of the spacetime. Under suitable assumptions on the $w_\Lambda$, we have both a non-extremal event horizon and a cosmological horizon, and we consider the equation in the region bounded by them. The argument in this case is similar, as for $\Lambda$ sufficiently small, the event horizon can be thought of as ``far" from the cosmological horizon and to have a small effect on quasinormal modes supported there. The frequencies approximated by $-\sqrt{\Lambda/3}\mathbb{N}$ can be thought of as the de Sitter modes of \cite{hintz2021quasinormalsds}, where this phenomenon is discussed for the Schwarzschild-de Sitter black hole.\\\\
In \cite{hintz2021quasinormalsds}, the authors use the spherical symmetry of the spacetime to project the Klein-Gordon equation onto fixed angular momenta and proceed to obtain uniform estimates on the resulting families of ordinary differential equations in suitable weighted spaces. By further using the dual resonant states \cite{hintz2021quasinormal} (or co-modes) of the problem, they set up a Grushin problem to obtain results on the convergence of the frequencies and modes to those of de Sitter in the limit as the black hole mass goes to 0. In this limit, the high frequency modes described in \cite{sabarretozworski} leave any bounded set, so one can show that these zero-damped modes are the only ones present. The results of \cite{hintz2021quasinormalsds} were recently extended to Kerr-de Sitter spacetimes in \cite{hintz2021kds} using more powerful techniques. In this paper, we only show existence of a sequence of zero-damped modes using Gohberg-Sigal theory and the concentration of the co-modes on the cosmological horizon, however we do not require restriction to fixed angular momenta and we consider a more general class of spherically symmetric spacetimes.\\\\
The final result is an application of the above to a Reissner-Nordstr\"om-de Sitter black hole. We consider the conformal Klein-Gordon equation again and take the extremal limit for the event horizon instead of the cosmological one. We make use of a conformal transformation to swap the event and cosmological horizons and simply apply the above result to obtain the following:
\begin{repthm}{thmrnds}[Rough version]
Consider the conformal Klein-Gordon equation on a Reissner-Nordstr\"om-de Sitter black hole background:
\begin{align*}
    -\Box_g\psi+\frac{R}{6}\psi=0.
\end{align*}
In the extremal limit where the event and Cauchy horizons coalesce, this equation exhibits the phenomenon of zero-damped quasinormal frequencies.
\end{repthm}
This paper is organised as follows. In \cref{review}, we review the definition of quasinormal modes and apply this to \eqref{CKG} with $V=0$. We can compute explicitly the spectrum and modes in this case, demonstrating the existence of zero-damped modes for this equation. We then review some of the literature which discusses this phenomenon.\\\\
In \cref{potentialstability}, we consider the effect of various classes of potential as perturbations to the equation. We begin by considering the inverse square potential (where we can compute the frequencies and modes explicitly) and establish existence of zero-damped frequencies in this case. We proceed similarly for constant potentials and then use the idea of co-modes to prove existence for compactly supported ones (this is \cref{potentialcomodes}). Finally we begin to establish the result for perturbations due to a potential using Gohberg-Sigal theory. We find that for spherically symmetric potentials we can obtain an expansion for the frequencies in terms of the size of the perturbation (\cref{seriesexpansion} combined with \cref{extrlmt}) and finally that for suitably decaying potentials, we still have the phenomenon of zero-damped modes in the extremal limit (\cref{mainresultthm}).\\\\
In \cref{blackhole}, we apply the above methods to a spherically symmetric black hole spacetime with suitable decay conditions on the metric components (\cref{mainresultthmblackhole}). The analysis is very similar, however there are subtleties pertaining to comparing operators defined on different domains and the estimates take a little more work to obtain. Finally we apply this method to the conformal Klein Gordon equation on a Reissner-Nordstr\"om-de Sitter black hole background to prove \cref{thmrnds}.\\\\
Throughout, we shall use the metric sign convention $(-,+,+,+)$ and use geometric units where $G=c=1$.
\section{Quasinormal modes in de Sitter}\label{review}
Quasinormal modes are similar to normal modes in that they demonstrate that there is a period of time in which certain characteristic modes of oscillation dominate the evolution of some perturbation. They also exhibit key differences: they are exponentially damped and appear only in a limited time interval (in the asymptotically flat case) whereas normal modes can propagate from arbitrarily early to arbitrarily late times. In the case of a system with normal modes, one usually imposes boundary conditions which make a suitable operator self-adjoint on a function space of interest and exploits this to get a discrete spectrum with an orthogonal eigenbasis of the function space. The situation is more complicated when considering systems where quasinormal modes arise and there are multiple approaches to defining them. We shall outline two of these in the following sections.
\subsection{The traditional approach}
			When defining quasinormal modes, the coordinates used are traditionally with respect to a static slicing which leads to an equation involving a self-adjoint operator. We work on the static patch of de Sitter so the spacetime manifold is $[0,\infty)\times B_{1/\kappa}$, where $B_{1/\kappa}:=\{x\in\mathbb{R}^3\ |\ \norm{x}< 1/\kappa \}$ is the open ball of radius $1/\kappa$, and we use coordinates $(t,r,\theta,\phi)$ so the metric is
			\begin{align}
			    g=-F(\kappa r)dt^2+\frac{dr^2}{F(\kappa r)}+r^2g_{S^2}\label{metric1}
			\end{align}
			where $F(r)=1-r^2$ and $g_{S^2}$ is the metric on the unit sphere. We shall consider the initial value problem
		\begin{align}
		    \Box_g \psi -2\kappa^2\psi=0, \quad \psi|_{t=0}=\psi_0, \quad \partial_t\psi|_{t=0}=\psi_1,\label{CKG1}
		\end{align}
		where $\Box_g$ is the wave operator and $\kappa$ is the surface gravity of the cosmological horizon. This is related to the cosmological constant of the spacetime: $\Lambda=3\kappa^2$. We introduce the tortoise coordinate $r_*$ defined by the equation $\kappa r_*=\artanh(r/\kappa)$ which, after dividing through by $\cosh^2(r_*)$, gives the equation
\begin{align}
		    -\frac{\partial^2 \psi}{\partial t^2}+\frac{\partial^2 \psi}{\partial r_*{}^2}+\frac{4}{\sinh(2r_*)}\frac{\partial \psi}{\partial r_*}+\frac{\slashed{\Delta}\psi}{\sinh^2(r_*)}-\frac{2}{\cosh^2(r_*)}\psi
		    =0.\label{CKGstatic1}
		\end{align}
Here $\slashed{\Delta}$ is the Laplace-Beltrami operator on the unit sphere, $S^2$. We can then define $\Psi=\psi\tanh (r_*)$ and decompose into spherical harmonics to simplify the problem to a two dimensional wave equation for each angular momentum sector:
		\begin{align}
		    -\frac{\partial^2 \Psi_{lm}}{\partial t^2}+\frac{\partial^2 \Psi_{lm}}{\partial r_*{}^2}+V_{lm}(r_*)\Psi_{lm}
		    =0\label{CKGstatic2}
		\end{align}
		where $V_{lm}(r_*)=-l(l+1)/\sinh^2r_*$.
We see that this is of the form
\begin{align*}
    \left(-\frac{\partial^2}{\partial t^2}+H\right)\Psi_{lm}=0
\end{align*}
for some $H$, a self-adjoint operator between suitably chosen spaces. We then Laplace transform in time to obtain the equation
\begin{align}
    (H-s^2)\hat{\Psi}_{lm}=f,\label{laplace}
\end{align}
where $f$ is constructed from the initial data. If we can solve the above equation in some region $\realpart(s)>c_0$, then we can use the Bromwich inversion formula to solve the full initial value problem \eqref{CKG1}. To do so, we may attempt to construct a Green's function. We see from \eqref{CKGstatic2} that the potential has exponential decay, so we would expect smooth solutions which obey
\begin{align*}
    \hat{\Psi}_{lm}^{\pm} \sim e^{\pm s r_*} \quad \text{as}\quad r_*\rightarrow \infty.
\end{align*}
Assuming $\realpart s>0$ for the moment, we see one of these is decaying. Considering a Taylor series about $r_*=0$, we also expect a solution which obeys $\hat{\Psi}_{lm}^0(r_*)=O(r_*^{l+1})$ as $r_*\rightarrow 0$. Now suppose these are not linearly independent for some $s\in\mathbb{C}$. Then $\hat{\Psi}_{lm}^0\in L^2(0,\infty)$ since it decays rapidly and is well behaved at 0. So we can compute
\begin{align*}
    \realpart\int_0^\infty -(\overline{s\hat{\Psi}_{lm}^0})\left(H-s^2\right)\hat{\Psi}_{lm}^0 dr_*=\realpart(s)\int_0^{\infty}\left(|s|^2|\hat{\Psi}_{lm}^0|^2-\overline{\hat{\Psi}_{lm}^0}H\hat{\Psi}_{lm}^0\right)dr_*.
\end{align*}
Noting that $H$ is self-adjoint so the second term is real. Since $\hat{\Psi}_{lm}^0$ obeys the equation, the above quantity vanishes. However, if $\realpart(s)>\norm{H}_{H^2\rightarrow L^2}>0$ (the operator norm of $H: H^2\rightarrow L^2$), we see
\begin{align*}
    0=\realpart(s)\int_0^{\infty}\left(|s|^2|\hat{\Psi}_{lm}^0|^2-\overline{\hat{\Psi}_{lm}^0}H\hat{\Psi}_{lm}^0\right)dr_*\ge \realpart(s)\left(|s|^2-\norm{H}_{H^2\rightarrow L^2}\right)\norm{\hat{\Psi}_{lm}^0}_{L^2}>0,
\end{align*}
a contradiction. So there exists a constant $c_0$ such that for $\realpart(s)>c_0$, the solutions outlined above are linearly independent and we can construct a Green's function
    \begin{align*}
		    G(r_*,\xi; s)=
		    \begin{cases}
		    \frac{\hat{\Psi}_{lm}^0(r_*)\hat{\Psi}_{lm}^{-}(\xi)}{W(s)} & 0\le r_*<\xi<\infty\\
		    \frac{\hat{\Psi}_{lm}^0(\xi)\hat{\Psi}_{lm}^{-}(r_*)}{W(s)} & 0\le\xi<r_*<\infty
		    \end{cases},
	\end{align*}
	where $W(s):=\partial_{r_*}\hat{\Psi}_{lm}^0\hat{\Psi}_{lm}^{-}-\hat{\Psi}_{lm}^0\partial_{r_*}\hat{\Psi}_{lm}^{-}$ is the Wronskian (which is independent of $r_*$ in this case). It can be shown that the Green's function is holomorphic in $s$ for this region (this will follow from subsequent arguments in this paper) and moreover we can analytically continue to a larger domain in $\mathbb{C}$ by noting that the only obstruction to defining the above are the places where $W(s)=0$. The values of $s$ for which this holds are isolated as $W$ is also holomorphic (it is constructed from solutions to an equation which depends analytically on $s$) and we can define these to be quasinormal frequencies. We can define quasinormal modes as non-trivial solutions $\hat{\Psi}_{lm}^0$ to \eqref{laplace} with $f=0$ obeying the boundary conditions $\hat{\Psi}_{lm}^0(0)=0$ and $\hat{\Psi}_{lm}^0(r_*)\sim e^{-s r_*}$ as $r_*\rightarrow 0$. Note that for $\realpart s <0$, this asymptotic boundary condition is ambiguous but can be overcome through the method of complex scaling \cite{sabarretozworski}.\\\\
	We can construct a solution to the original equation given a quasinormal frequency $s$ and the corresponding mode $\hat{\Psi}_{lm}^0$ by setting
	\begin{align*}
	    \Psi_{lm}(t,r_*)=e^{st}\hat{\Psi}_{lm}^0(r_*)\sim e^{s(t-r_*)} \quad \text{as} \quad r_*\rightarrow \infty.
	\end{align*}
	This is asymptotic to a right-moving wave, so the boundary conditions imposed on $\hat{\Psi}_{lm}^0$ are usually called `outgoing' boundary conditions.
\subsection{Regularity quasinormal modes}\label{reg}
The discussion above can be made rigorous with the method of complex scaling (see \cite{sabarretozworski}). However, an approach based on this time foliation for which energy is conserved makes it harder to discuss dispersive decay. A null or hyperboloidal slicing of the spacetime would build the idea of dispersion into the problem since energy naturally radiates through the horizons in this case. This can be seen using the results of \cite{redshift} for asymptotically de Sitter spacetimes and \cite{rp} for ones which are asymptotically flat. The limitations of the traditional definition were first overcome through microlocal methods by Vasy in \cite{vasy}, however one can also use physical space methods to recover many of his results. This was outlined by Warnick in \cite{cmw}, which also contains a detailed comparison with the traditional approach. We shall apply these physical space methods to the specific case of the Klein-Gordon equation in de Sitter. In order to do so, we need to use coordinates regular at the horizon. The following change of co-ordinates gives a space-like foliation of the spacetime (see \cref{fig:slicing}):
\begin{align*}
		    \tau&=t+\frac{1}{2\kappa}\log(1-\kappa^2r^2),\\
		    \rho&=\kappa r.
		\end{align*}
		Using the co-ordinate $\tau$ gives a foliation where the leaves intersect the horizon and are regular there (see \cref{fig:slicing})
		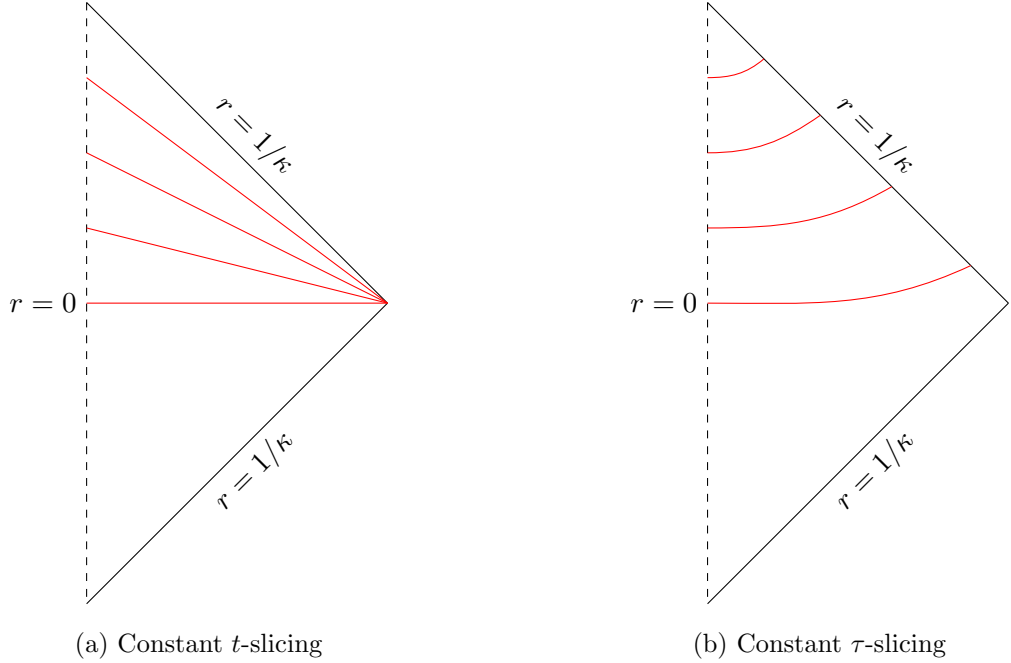
\begin{figure}
    \centering
    \begin{subfigure}[b]{0.5\textwidth}
    \centering
    \begin{tikzpicture}
\node (A) at ( 2,0){};
\node (B) at (-2,4){};
\node (C) at (-2,-4){};
\draw (A.center)--node[midway, above, sloped] {$r=1/\kappa$} (B.center);
\draw (A.center)--node[midway, below, sloped] {$r=1/\kappa$}(C.center);
\draw [dashed] (B.center)--node[midway,left] {$r=0$} (C.center);
\draw [red] (A.center)--(-2,0);
\draw [red] (A.center)--(-2,1);
\draw [red] (A.center)--(-2,2);
\draw [red] (A.center)--(-2,3);
\end{tikzpicture}
\caption{Constant $t$-slicing}
    \end{subfigure}%
    \begin{subfigure}[b]{0.5\textwidth}
    \centering
    \begin{tikzpicture}
\node (A) at ( 2,0){};
\node (B) at (-2,4){};
\node (C) at (-2,-4){};
\draw (A.center)--node[midway, above, sloped] {$r=1/\kappa$} (B.center);
\draw (A.center)--node[midway, below, sloped] {$r=1/\kappa$}(C.center);
\draw [dashed] (B.center)--node[midway,left] {$r=0$} (C.center);
\draw [red] (-2,0) to [out=0, in=205] (1.5,0.5);
\draw [red] (-2,1) to [out=0,in=210] (0.45,1.55);
\draw [red] (-2,2) to [out=0,in=215] (-0.5,2.5);
\draw [red] (-2,3) to [out=0,in=220] (-1.25,3.25);
\end{tikzpicture}
\caption{Constant $\tau$-slicing}
    \end{subfigure}
    \caption{Penrose diagrams of the static patch of de Sitter depicting (a) the static slicing and (b) the hyperboloidal slicing}
    \label{fig:slicing}
\end{figure}
		After this change of coordinates, \eqref{metric1} becomes
		\begin{align}
		    g_\kappa=-F(\rho)d\tau^2-\frac{2\rho}{\kappa}d\tau d\rho+\frac{1}{\kappa^2}\left(d\rho^2+\rho^2g_{S^2}\right)
		\end{align}
		and can be extended to the boundary of the open ball so it is defined on $[0,\infty)\times \overline{B_1}$. We further change from polar coordinates on $\overline{B_1}$ to Cartesian ones $\{x_i\}_{i=3}^3$ such that $\sum_{i=1}^3(x_i)^2=\rho^2$. We can write \eqref{CKG1} as:
		\begin{align}
		    -\kappa^2 \sum_{i=1}^3\sum_{j=1}^3a_{ij}\partial_i\partial_j\psi+ 4\kappa^2 \sum_{i=1}^3x_i\partial_i\psi+2\kappa^2\psi+2\kappa \sum_{i=1}^3x_i\partial_i\partial_\tau\psi+3\kappa\partial_\tau\psi+\partial_\tau^2\psi=0, \label{CKG2}
		\end{align}
		where $a_{ij}=\delta_{ij}-x_ix_j$. We impose initial conditions at $\tau=0$ now instead: $\psi(0,\textbf{x})=\psi_0(\textbf{x})$, $\partial_\tau \psi(0,\textbf{x})=\psi_1(\textbf{x})$. Setting $(u,v)=(\psi,\partial_\tau\psi)$, we can recast the problem to the following form:
\begin{align}
	\frac{\partial}{\partial\tau}\begin{pmatrix}u\\v\end{pmatrix}=\begin{pmatrix}0&1\\ -\kappa^2L_0 & -\kappa P\end{pmatrix}\begin{pmatrix}u\\v\end{pmatrix}, \quad \left.\begin{pmatrix}u\\v\end{pmatrix}\right|_{\tau=0}=\begin{pmatrix}\psi_0\\\psi_1\end{pmatrix},\label{hyperbolic}
\end{align}
where
\begin{align}
	L_{0}u&:=-\sum_{i=1}^3\sum_{j=1}^3a_{ij}\partial_i\partial_ju+ 4\sum_{i=1}^3x_i\partial_iu+2u,\label{defineL}\\
	Pu&:=\sum_{i=1}^3x_i\partial_iu+\frac{3}{2} u.
\end{align}
By the theory developed in \cite{cmw}, we have the following result:
\begin{prop}\label{semigroup}
	For each $\tau\ge 0$, define the operator $\mathcal{S}(\tau): H^k(B_1)\times H^{k-1}(B_1)\rightarrow H^k(B_1)\times H^{k-1}(B_1)$ such that
	\begin{align*}
		\mathcal{S}(\tau)\begin{pmatrix}\psi_0\\\psi_1\end{pmatrix}=\begin{pmatrix}u\\v\end{pmatrix}(\tau),
	\end{align*}
	where $(u,v)$ is the unique solution to \eqref{hyperbolic}. Then the family of operators $\{\mathcal{S}(\tau)\}_{\tau\ge0}$ forms a $C^0$-semigroup acting on $H^k(B_1)\times H^{k-1}(B_1)$. The infinitesimal generator of $\mathcal{S}$ is the closed, densely defined operator $\mathcal{A}: D^k(\mathcal{A})\rightarrow H^k(B_1)\times H^{k-1}(B_1)$ given by
	\begin{align*}
		\mathcal{A}=\begin{pmatrix}0&1\\ -\kappa^2L_0 & -\kappa P\end{pmatrix},
	\end{align*}
	where $D^k(\mathcal{A})=\{\psi\in H^k(B_1)\times H^{k-1}(B_1) \mid \mathcal{A}\psi \in H^k(B_1)\times H^{k-1}(B_1)\}$. The resolvent $(\mathcal{A}-s)^{-1}: H^k(B_1)\times H^{k-1}(B_1) \rightarrow H^k(B_1)\times H^{k-1}(B_1)$ is well-defined and holomorphic on $\realpart(s)>C$ for some real constant $C$.
\end{prop}
Note $H^k(B_1)$ denotes the usual Sobolev spaces and we take the case $k=0$ to mean $L^2(B_1)$. We see that eigenvectors of $\mathcal{A}$ must have components obeying:
\begin{align*}
	su&=v,\\
	sv&=-\kappa^2 L_0u-\kappa Pv,
\end{align*}
which means they are of the form
\begin{align}
	\mathbf{u}=\begin{pmatrix}u\\ su\end{pmatrix},\label{qnmvector}
\end{align}
where $\kappa^2L_0u+2s\kappa Pu+s^2u=0$. We define the densely defined operator:
\begin{align*}
	L_s:=L_0+2sP+s^2
\end{align*}
and note that there is a one-to-one correspondence between eigenvectors of $\mathcal{A}$ and solutions to $L_{s/\kappa}u=0$ with $u\in D^k(L_{s/\kappa})$. While the dependence on $s$ of this second problem is slightly more complicated, it is an elliptic problem (away from the cosmological horizon where ellipticity degenerates) and we can apply Fredholm theory to it.\\\\
For $\realpart s>1/2-k$, the $L_s$ form a holomorphic family of Fredholm operators $D^k(L_s)\rightarrow H^{k-1}(B_1)$ where we define $D^k(L_s)$ to be the domain of $L_s$ i.e. the closure with respect to the graph norm of
\begin{align}
	\left\{u\in C^{\infty}(B_1)\,\middle\vert\, \norm{u}_{D^k}:=\norm{u}_{H^{k-1}}+\norm{L_su}_{H^{k-1}}<\infty\right\}.
\end{align}
Note that for the range of values of $s$ we are considering, this set is independent of the value of $s$ and is in fact a subset of $H^k(B_1)$. We can use the redshift effect \cite{redshift} at the horizon to establish the following result.
\begin{prop}\label{Fredholm}
	Let $f\in H^{k-1}(B_1)$ and $\real(s)>1/2-k$. If $L_su=f$,
	we have either:
	\begin{enumerate}
		\item[(i)] there exists a unique solution to $L_su=f$ where $u\in H^k(B_1)$
		\item[(ii)] there exists a finite dimensional space of solutions $v\in C^{\infty}(\overline{B_1})$ to $L_sv=0$. Moreover this can only occur at isolated values of $s$.
	\end{enumerate}
\end{prop}
The result follows from the theory established in \cite{cmw}: once estimates corresponding to the redshift effect and the time-like isometry have been established, we are left with a Fredholm problem. The upshot is that given $f\in H^{k-1}(B_1)$, we have a family of solutions $u(s)\in H^k(B_1)$ which are meromorphic in $s$ for $\realpart(s)>1/2-k$ and that the locations of the poles of $u(s)=L_s^{-1}f$ are quasinormal frequencies. It is important to note here that for any $s$ to the left of the region defined (i.e. with $\realpart(s)<1/2-k$), there exists $v\in H^k(B_1)$ such that $L_sv=0$, so we must restrict to this subset of $\mathbb{C}$.  We can now make the following definition:
\begin{defn}\label{defqnm}
	With all the notation as above, we say $s\in \{s\in\mathbb{C} \mid\realpart(s)>1/2-k\}$ is a \textit{quasinormal frequency} if it is a pole of the resolvent $(\mathcal{A}-s)^{-1}:H^k(B_1)\times H^{k-1}(B_1)\rightarrow H^k(B_1)\times H^{k-1}(B_1)$. The corresponding quasinormal modes are vectors $\mathbf{u}\in D^k(\mathcal{A})$ such that $\mathcal{A}\mathbf{u}=s\mathbf{u}$.
\end{defn}
The approach outlined above will work for any subextremal horizon with anti-de Sitter or de Sitter ends which is what we are primarily interested in. There are problems in the case that the horizon is extremal, which were outlined and addressed for a model problem in \cite{modelproblem} and for Reissner-Nordstr\"om in \cite{cmwdgqnm}. One of the key reasons this method fails for extremal horizons is that we can no longer exploit the redshift effect at the horizon to establish the necessary estimates for a Fredholm alternative. We see this from the fact that for initial data of given regularity (say $H^{k-1}(B_1)$), the region on which we can define quasinormal frequencies in the left half-plane becomes smaller as $\kappa\rightarrow 0$. To retain any sensible notion of invertibility of $L_s$ in this region, we must restrict to at least $C^{\infty}(\overline{B_1})$ functions. Even this is not enough: one can construct non-trivial smooth functions solving the homogeneous problem (see \cite{modelproblem}). The problem of which space to use is resolved in \cite{modelproblem, cmwdgqnm} with the construction of suitable $L^2$-based Gevrey spaces.\\\\
We shall now compute the quasinormal frequencies for a conformally coupled Klein-Gordon field in de Sitter. Note that to determine the corresponding modes, it suffices to determine the scalar function $u$ in \eqref{qnmvector} since the second component is just a multiple of it. To aid our computation, we use the spherical symmetry of de Sitter to decompose $u$ into spherical harmonics:
\begin{align*}
    u=\sum_{l=0}^{\infty}\sum_{m=-l}^lu_{lm}(\rho)Y_{lm}(\theta,\phi),
\end{align*}
where we have switched back to spherical polar coordinates. Using the orthogonality of the $Y_{lm}$, we can separate the equations to reduce the problem to a set of ordinary differential equations on $[0,1]$:
\begin{align*}
		    (1-\rho^2)\partial_\rho^2u_{lm}+\left(\frac{2}{\rho}-4\rho-2\s \rho\right)\partial_\rho u_{lm}-\left(\frac{l(l+1)}{\rho^2}+\s^2+3\s+2\right)u_{lm}=0,
\end{align*}
where $\s:=s/\kappa$. This is a second order Fuchsian equation with four regular singular points: $0$, $\pm 1$ and $\infty$. Considering the indicial equation at $\rho=0$, we see that the roots are $l$ and $-l-1$. The second exponent will clearly lead to solutions which are not smooth at the origin and so can be discarded. We seek a solution of the form $u_{lm}(\rho)=\rho^lv_{lm}(\rho)$ which leads to the equation
\begin{align*}
		    (1-\rho^2)\partial_\rho^2v_{lm}+\left(\frac{2l+2}{\rho}-(2\s+2l+4)\rho\right)\partial_\rho v_{lm}-(\s+l+1)(\s+l+2)v_{lm}=0.
		\end{align*}
		Changing variable to $z=\rho^2$, we have:
		\begin{align*}
		    z(1-z)\partial_z^2v_{lm}+\left(l+\frac{3}{2}-\left(\s+l+\frac{5}{2}\right)z\right)\partial_zv_{lm}-\frac{\s+l+1}{2}\cdot \frac{\s+l+2}{2}v_{lm}=0.
		\end{align*}
		This is the hypergeometric equation, which gives us a unique solution that is smooth at $\rho=0$:
		\begin{align*}
		    u_{lm}(\rho)=\rho^l{}_2F_1\left[\frac{\s+l+1}{2},\frac{\s+l+2}{2};\frac{3}{2}+l; \rho^2\right],
		\end{align*}
		where ${}_2F_1$ is the hypergeometric function as defined in \cite{hypergeometricfunction}. Thus we have found candidates for quasinormal modes: vectors constructed from functions of the form
		\begin{align*}
		    u(\rho, \theta, \phi)=\rho^lY_{lm}(\theta,\phi)\cdot {}_2F_1\left[\frac{\s+l+1}{2},\frac{\s+l+2}{2};\frac{3}{2}+l; \rho^2\right]
		\end{align*}
		satisfy the appropriate equation, we simply need to check that they are smooth at the horizon. Since $\rho^l Y_{lm}(\theta, \phi)$ is smooth in $\overline{B_1}$, it suffices to check smoothness of the hypergeometric function given above on $[0,1]$.\\\\
		From standard results for Fuchsian equations and the Taylor series about $\rho=0$ of $v_{lm}$, the solution is analytic on the unit disc, so we only need to check behaviour at $\rho=1$.	Provided $\s\notin \{-l-1, -l-2, \dots\}$, the radius of convergence of the Taylor series about $\rho=0$ for this hypergeometric function is 1 and hence there must be a singularity of the function on the unit circle in $\mathbb{C}$. This can only occur at a regular singular point of the equation, namely $0,\pm 1,\infty$. Hence we have a singularity at either $\rho=1$ or $\rho=-1$, but since $v_{lm}(\rho)$ is even, there must be one at both. Whether this singularity arises from a pole or a branch point, after sufficiently many derivatives, $\partial_\rho^ku_{lm}$ will not be continuous on $[0,1]$ and hence we cannot have quasinormal frequencies for $\s\notin \{-l-1,-l-2,\dots\}$. Now it suffices to check that elements of this set are indeed quasinormal frequencies: for these values of $\s$, the Taylor series for $v_{lm}$ about $\rho=0$ terminates. This gives us the following expressions for $u$:
		\begin{equation}
		    \begin{split}
    	    \s=-l-2n-1,\quad &\quad u=\rho^lY_{lm}(\theta,\phi)\sum_{k=0}^{n}\frac{n(n-1/2)\dots (n-k+1)(n-k+1/2)}{\left(\frac{3}{2}+l\right)_kk!}{\rho^{2k}},\\
		    \s=-l-2n-2,\quad &\quad u=\rho^lY_{lm}(\theta,\phi)\sum_{k=0}^n\frac{n(n+1/2)\dots (n-k+1)(n-k+3/2)}{\left(\frac{3}{2}+l\right)_kk!}{\rho^{2k}},
		\end{split} \label{qnm}
		\end{equation}
		with $n\in\mathbb{N}_0:=\mathbb{N}\cup \{0\}$. We have used the notation $(a)_k$ for the Pochhammer symbol i.e. $(a)_0=1$ and $(a)_k=a(a+1)\dots (a+k-1)$ for $a\in\mathbb{C}$ and $k\in\mathbb{N}$. These are polynomials and thus clearly smooth, so the quasinormal frequencies for a given angular momentum sector $l$ are $\{-l-1, -l-2, \dots\}$ with the corresponding modes determined by the above functions. Returning to Cartesian coordinates, we see that for each frequency $\s=-n$ there is a finite dimensional subspace of analytic solutions (spanned by $n^2$ polynomials in $\{x_1,x_2,x_3\}$) as expected from the results in \cite{analyticqnm}.
\subsection{Quasinormal co-modes}\label{comodesection}
From \cref{Fredholm}, we know that $L_{\s}: H^k\rightarrow H^{k-1}$ is a holomorphic family of Fredholm operators on $\Omega=\{\s\in\mathbb{C} \mid\realpart(\s)>1/2-k\}$. We know that the inverse exists at some point in $\Omega$, therefore we have a meromorphic family of operators $L^{-1}_{\s}$ on $\Omega$ from Theorem C.8 in \cite{zworski}. This means that given $\s_0\in \Omega$, we can write:
		\begin{align*}
		    L^{-1}_{\s}=A_0(\s)+\sum_{j=1}^J\frac{A_{-n_j}}{(\s-\s_0)^{n_j}},
		\end{align*}
		where the $A_{-n_j}$ are finite rank operators and $A_0(\s)$ is a holomorphic family of Fredholm operators. In light of Theorem C.10 from \cite{zworski} applied to $2(P+s)L_s^{-1}$ and considerations of eigenvalues and eigenvectors in this case, we see that $L_s^{-1}$ only has simple poles (this is discussed in more detail in \cref{polesdiscussion}). We can then use the fact that
		\begin{align*}
		    (s-s_0)L_{s}^{-1}L_{s}u=(s-s_0)L_{s}L_{s}^{-1}u=(s-s_0)u\rightarrow 0
		\end{align*}
		to deduce that
		\begin{align*}
		    \img A_{-1} = \ker L_{s_0}, \quad \quad \img L_{s_0}= \ker A_{-1}.
		\end{align*}
		So the residue projects onto the space of solutions to $L_su=0$, which we know is finite dimensional. Hence if the solutions to $L_su=0$ are spanned by $\{w_j\}_{j=1}^N$, we can write
		\begin{align}
		    A_{-1}=\sum_{j=1}^N w_j\theta_j \label{projection}
		\end{align}
		where the $\theta_i$ are continuous linear functionals $H^k(B_1)\rightarrow \mathbb{C}$ which vanish on $\img L_{s_0}$. If we further have that $A_{-1}$ is a projection, then we can also show that $\theta_i(w_j)=\delta_{ij}$. This gives motivation to think of this problem in the dual picture. This idea was discussed in \cite{hintz2021quasinormal} where the authors defined distributions obeying certain conditions as \textit{dual resonant states}. The discussion below provides an equivalent definition.\\\\
		We note that $L_s: D^{k+1}(L_s) \rightarrow H^{k}(B_1)$ is a bounded linear operator and we can define the adjoint $L_s^{\dagger}: H^{k}(B_1)\rightarrow D^{k+1}(L_s)$ with respect to the $H^k(B_1)$ inner product. Since both $L_s$ and $L_s^{\dagger}$ are Fredholm of index zero, we can deduce that their kernels have the same dimension and so we seek to define a notion of `co-mode' using the kernel of $L_s^{\dagger}$.\\\\
		Another motivation is the expression of $A_{-1}$ given in \eqref{projection}. We know the $w_i$ are the quasinormal modes, so we seek an understanding of the $\theta_i$. Since it is often easier to work with distributions than objects in the continuous dual of $H^k(B_1)$, we shall outline a few identifications between spaces. By virtue of the Riesz representation theorem, $H^k(B_1)$ is isomorphic to its continuous dual $H^k(B_1)'$, and we can identify the kernel of $L_s^{\dagger}$ with the subspace of $H^k(B_1)'$ that vanishes on the image of $L_s$. Fixing some small $\epsilon>0$, we can also identify $H^k(B_1)'$ with the space
		\begin{align*}
		    X=\left\{ \theta \in H^{-k}(B_{1+\epsilon})\,\middle\vert\, \supp \theta \subset B_1\right\},
		\end{align*}
		where we say $\supp\theta\subset B_1$ if for any $u\in H^k(B_{1+\epsilon})$ with $\supp u\subset B_{1+\epsilon} \setminus B_1$, $\theta(u)=0$. Since $X \subset H^{-k}(B_{1+\epsilon})\subset \mathcal{D}'(B_{1+\epsilon})$, we can seek distributional solutions to the adjoint problem $L_s^{\dagger}\theta =0$.
		\begin{defn}
		We say $\theta\in\mathcal{D}'(B_{1+\epsilon})$ is a \textit{quasinormal co-mode} if it satisfies:
    	\begin{enumerate}
	    \item[(i)] there exists $C>0$ such that $|\theta(u)|<C\norm{u}_{H^k(B_1)}$ for all $u\in C^{\infty}_c(B_{1+\epsilon})$ ($\Rightarrow \supp\theta\subset B_1$)
	    \item[(ii)] $L_s^{\dagger}\theta=0$ in the sense of distributions, i.e. $\theta(L_su)=0$ for any $u\in C^{\infty}_c(B_{1+\epsilon})$.
	    \end{enumerate}
		\end{defn}
		\noindent By noting that $L_s^{-1}L_s=I_{D^{k+1}(L_s)}$, we see that the co-modes defined above are precisely the $\theta_i$ appearing in \eqref{projection}. We can compute some of these explicitly for the problem in de Sitter. We define the following distributions on $B_{1+\epsilon}$: for any test function $f$, we set
		\begin{align}
		    \Theta_{lm}(f):=\int_{S^2}Y_{lm}f d\sigma,\label{comodes}
		\end{align}
		where $S^2$ is the unit sphere and $d\sigma$ is the standard measure on it. By virtue of the divergence theorem, we see this is a continuous linear map on $H^k(B_1)$ for $k\ge 1$ and can be extended continuously to this space. We define the derivative operator
		\begin{align}
		    \tilde{D}u:=\frac{1}{\rho}\left(\sum_{i=1}^3x^i\partial_iu+u\right)=\frac{1}{\rho}\partial_{\rho}(\rho u) \label{twisted}
		\end{align}
		and note that
	    \begin{align*}
	        L_su=-(1-\rho^2)\tilde{D}^2u+2(s+1)\rho\tilde{D}u+s(s+1)u-\frac{1}{\rho^2}\slashed{\Delta}u.
	    \end{align*}
	    Again, $\slashed{\Delta}$ is the Laplace-Beltrami operator on the unit sphere.
	    \begin{prop}\label{puredScomodes}
	        $\tilde{D}^k\Theta_{00}:H^{k+1}(B_1)\rightarrow \mathbb{C}$ is a co-mode for each $k\in\mathbb{N}_0$
	    \end{prop}
	    \begin{proof}
	    We see that for $u\in H^{k+2}(B_1)\subset D^{k+2}(L_s)$
	    \begin{align*}
	        \tilde{D}^k\Theta_{00}(L_su)=(-1)^k\Theta_{00}(\tilde{D}^kL_su).
	    \end{align*}
	    A simple computation shows that for $u\in C^{\infty}(\overline{B_1})$
	    \begin{align*}
	        \tilde{D}^kL_su=&-(1-\rho^2)\tilde{D}^{k+2}u+2(s+k+1)\rho\tilde{D}^{k+1}u\\&+(s+k)(s+k+1)\tilde{D}^{k}u-\slashed{\Delta}\tilde{D}^{k}\left(\frac{u}{\rho^2}\right).
	    \end{align*}
	    Using a suitable decomposition of $u$ into spherical harmonics (which exists for smooth functions), the fact that $\slashed{\Delta}Y_{00}=0$ and orthonormality of spherical harmonics, we are left with
	    \begin{align*}
	        \tilde{D}^k\Theta_{00}(L_su)=2(s+k+1)\int_{S^2}\tilde{D}^{k+1}ud\sigma+(s+k)(s+k+1)\int_{S^2}\tilde{D}^{k}ud\sigma.
	    \end{align*}
	    So for the quasinormal frequency $s=-k-1$, we see that $\tilde{D}^k\Theta_{00}(L_su)=0$ for $u$ in a dense subset of $D^{k+2}(L_s)$. Since the operator above is linear and continuous, we see that it vanishes for all $u\in D^{k+2}(L_s)$ and hence $\tilde{D}^k\Theta_{00}$ is a co-mode.
	    \end{proof}
	    Thus we see that for each quasinormal frequency, there is a co-mode associated to it that is concentrated only on the horizon. One can further show that all co-modes occur as linear combinations of $\tilde{D}^k\Theta_{lm}$ by considering the kernel of a suitable matrix, which agrees with results in \cite{hintz2021quasinormal} obtained through different methods. These linear combinations can be found by hand through computing the eigenvectors of the previously mentioned matrices. Since we are content to find a single sequence of zero-damped mode frequencies, it suffices to consider the co-modes associated with the $l=0$ angular momentum sector.\\\\
	    It is worth noting that similar arguments can be applied to the wave equation in de Sitter to see that the co-modes are concentrated on the horizon in this case as well (this also is observed in \cite{hintz2021quasinormal}). The reason this holds for these special masses in the Klein-Gordon equation (and in fact, \textit{only} these special cases) is that the quasinormal spectrum contains $-\kappa\mathbb{N}$ (see \cref{constant}). The concentration of these co-modes means that the associated modes are not excited by data supported away from the horizon and so the standard complex scaling approach would not identify the corresponding frequencies as quasinormal frequencies.
		\subsection{Zero-damped modes}
So far, we have found that the spectrum of quasinormal frequencies for a conformally coupled Klein-Gordon field in de Sitter is $-\kappa \mathbb{N}$ where $\kappa$ is the surface gravity of the cosmological horizon. We can define a notion of `extremality' for the cosmological horizon by considering the extremal limit for the event horizon of a Kerr or Reissner-Nordstr\"om spacetime. As we approach extremality, the surface gravity of the event horizon goes to zero - we simply apply this to the cosmological horizon. In the $\kappa\rightarrow 0$ limit, we see that the interior of the spacetime more closely resembles Minkowski. We observe that the frequencies converge to zero and cluster closer together as they do so. Since there are infinitely many frequencies on the real line, we can informally think of these as forming the branch cut observed when trying to define quasinormal modes in Minkowski. For this paper, quasinormal modes whose frequencies exhibit this behaviour are called zero-damped modes, although other authors have used different terminology.\\\\
This phenomenon has been observed in several spacetimes with horizons approaching extremality. For example, in \cite{Yang_2013, azimmerman}, the authors study the quasinormal modes of the Teukolsky equation for nearly extremal Kerr via the traditional approach. After decomposing the equation using the symmetries of Kerr, they use both a WKB analysis in the eikonal limit (where the angular frequency $l$ is large) and the method of matched asymptotic expansions (the set of solutions were first correctly derived using this method by Hod in \cite{hod1}) to find zero-damped modes. In this spacetime, there is a sequence of frequencies which cluster on the line $\imag s=-m\Omega_H$ where $m$ is the azimuthal mode number of the quasinormal mode and $\Omega_H$ is the horizon frequency instead of $\imag s=0$ seen in de Sitter. The key observations in \cite{Yang_2013, azimmerman} were the existence of these modes and the fact that the spectrum bifurcates: there is a critical value of $m/(l+1/2)$ below which there are both zero-damped modes and what the authors called damped modes. The latter are quasinormal modes whose frequencies converge to some complex number with non-zero real part in the extremal limit. Above this critical value, only the zero-damped modes persist. This property is related to the fact that the zero-damped modes are concentrated near the turning point in the potential associated with the horizon.\\\\
The more general case of nearly extremal Kerr-Newman black holes was discussed in \cite{hod2} for the slowly rotating case and more recently by Zimmerman and Mark for tractable fields in \cite{azimmerman2016} without that assumption. In the latter paper, the authors tackle the simplified model of the Dudley-Finley equation for general nearly extremal Kerr-Newman spacetimes and find approximate expressions for both zero-damped and damped modes for any value of $a$.Existence for general perturbations is supported by the numerical work provided in \cite{santos2015}. Zimmerman and Mark also consider gravito-electromagnetic perturbations of near-extremal Reissner-Nordstr\"om in \cite{azimmerman2016} and demonstrate the existence of zero-damped modes in this case.\\\\
	The case of a charged scalar field on a  Reissner-Nordstr\"om background was discussed in \cite{ hod6, hod5} and later in \cite{pzimmerman} using similar techniques to \cite{hod1, Yang_2013, azimmerman}. Since both situations contain a $U(1)$ symmetry (Kerr is axisymmetric, while the charged scalar field has a gauge symmetry), many results carry over from one of the cases to the other by simply swapping $q$ for $m$ where $q$ is the charge of the scalar field and $m$ is the azimuthal mode number of a Kerr perturbation. As such, the phenomenon of zero-damped modes arises in this situation as well.\\\\
	This phenomenon has also been observed on a  Reissner-Nordstr\"om-de Sitter background for both scalar \cite{destounis1, destounis2} and fermionic \cite{destounis3} fields. These results have been generalised to higher dimensions \cite{destounis4} and with non-minimally coupled scalar fields \cite{destounis5}.\\\\
	Using key common features found in the examples above, we propose the following definition to make the notion of zero-damped modes somewhat more mathematically precise:
	\begin{defn}\label{defnzdm}
	Let $(\mathcal{M}, g)_{\kappa}$ be a family of spacetimes with non-degenerate Killing horizons of surface gravity $0<\kappa\le \kappa_0$. Consider a partial differential equation on this background for which the notion of regularity quasinormal modes as described in the previous subsection is well-defined. We say that the equation exhibits the phenomenon of \textit{zero-damped quasinormal frequencies} if there exists a sequence of functions $\{s_n:(0,\kappa_0]\rightarrow \mathbb{C}\}_{n=1}^\infty$ such that:
	\begin{enumerate}
	\item $s_n(\kappa)$ is a quasinormal frequency for each $n\in\mathbb{N}, \kappa\in (0,\kappa_0]$,
	\item there exists $\alpha\in\mathbb{R}$ such that $s_n(\kappa)\rightarrow i\alpha$ as $\kappa \rightarrow 0$ for each $n\in\mathbb{N}$,
	    \item $\Re s_n(\kappa)\rightarrow -\infty$ as $n\rightarrow \infty$ for each $\kappa\in (0,\kappa_0]$.
	\end{enumerate}
	\end{defn}
	A natural conjecture from the above discussion of examples is that these zero-damped modes are generic to nearly extremal Killing horizons. A start to determining whether they exist for any such set-up is to consider perturbations to a simple example exhibiting this phenomenon, such as de Sitter, and use these results as a black box to tackle more complicated examples.
\section{Perturbations due to a potential}\label{potentialstability}
Let us now consider the perturbation caused by adding a potential $V$ to the equation:
\begin{align}
    -\Box_g\psi + 2\kappa^2\psi +V\psi =0. \label{potential}
\end{align}
In our chosen rescaled coordinates, this becomes:
\begin{align}
    L_0\psi+\frac{2}{\kappa}P\partial_\tau \psi +\frac{1}{\kappa^2}\partial_{\tau}^2\psi+W\psi=0,\label{rescaledequation}
\end{align}
where $W(\textbf{x})=\frac{1}{\kappa^2}V(\textbf{x}/\kappa)$. Through the same reasoning as before, we see that the quasinormal frequencies depend on the invertibility of $L_{\s}+W$. The methods used before can be used again in this case, however this time the regularity of solutions is limited by the (weak) differentiability of $W$ as well as that of the initial data.
\subsection{The inverse square potential}
        Let us consider in detail an inverse square potential. We consider the equation
        \begin{align*}
            -\Box_g\psi+2\kappa^2\psi+\frac{V_0}{r^2}\psi=0,
        \end{align*}
        where $V_0>0$. We can obtain similar results to \cref{semigroup,Fredholm} for this equation, however the singular behaviour of the potential at the origin results in modifications to our approach. We define a modified energy 
        \begin{align*}
            E_1(u):=\sum_{i=1}^3\sum_{j=1}^3\int_{B_1}a_{ij}\partial_i\overline{u}\partial_judx+\gamma \int_{B_1}|u|^2dx+\int_{B_1}\frac{V_0|u|^2}{\rho^2}dx,
        \end{align*}
        which is well-defined for $u\in H^1(B_1)$ by virtue of Hardy's inequality. We can use this and the redshift effect at the horizon to obtain a result analogous to \cref{Fredholm}, however we lose the ability to differentiate the equation at the origin due to the singularity in the potential. The best we can expect is a solution in $H^1(B_1)$. Since the origin is the only obstacle to obtaining higher regularity (and we are only really interested in behaviour near the horizon, see \cite{cmwdgqnm}), we can remedy this by introducing a cut-off $\chi$ and splitting the problem into the two regions. We now consider a coupled system of partial differential equations with solutions in two different spaces ($H^1(\supp\chi)$ for the equation near the origin and $H^k(\supp(1-\chi))$ for the equation near the horizon). The methods for the degenerate elliptic problems described in \cite{cmw} can be used with minor adaptations to establish the appropriate Fredholm properties of this coupled system and obtain the following result:

        \begin{prop}\label{inversesquareFredholm}
	Let $f\in H^{k-1}(B_1)$ and $\realpart(s)>1/2-k$. If $L_su+V_0 u/\rho^2=f$,
	we have either:
	\begin{enumerate}
		\item[(i)] there exists a unique solution to $L_su+V_0 u/\rho^2=f$ where $u\in H^1(B_1)\cap H^k(B_1\setminus B_{2/3})$
		\item[(ii)] there exists $v\in H^1(B_1)\cap C^{\infty}(\overline{B_1\setminus B_{2/3}})$ and which obeys $L_sv=0$. Moreover this can only occur at isolated values of $s$.
	\end{enumerate}
        \end{prop}
          We can use this to establish estimates to get a result similar to \cref{semigroup} and define quasinormal modes and frequencies. To actually compute these, we may again use spherical symmetry to decompose into spherical harmonics and reduce the problem to an ordinary differential equation. The resulting equation is a second order Fuchsian equation with the same ordinary points:
          \begin{align*}
		    (1-\rho^2)\partial_\rho^2u_{lm}+\left(\frac{2}{\rho}-4\rho-2\s \rho\right)\partial_\rho u_{lm}-\left(\frac{l(l+1)+V_0}{\rho^2}+\s^2+3\s+2\right)u_{lm}=0.
    \end{align*}
    We can use precisely the same transformations as before, however the potential has changed the exponents at $\rho=0$ to
    \begin{align*}
        \sigma_l^{\pm}=-\frac{1}{2}\pm \sqrt{\left(l+\frac{1}{2}\right)^2+V_0}.
    \end{align*}
    The solution with exponent $\sigma_l^{-}$ at the origin gives a singular solution and thus through exactly the same methods as before see that solutions are:
    \begin{align*}
        u_{lm}(\rho)=\rho^{\sigma_l^+}{}_2F_1\left[\frac{\s+\sigma_l^++1}{2},\frac{\s+\sigma_l^++2}{2};\frac{3}{2}+\sigma_l^+; \rho^2\right].
    \end{align*}
    Since we require solutions which are smooth at the horizon, we see that quasinormal frequencies are of the form $\s=-\sigma_l^+-k$ for $k\in\mathbb{N}$ and the corresponding modes are polynomials multiplied by $\rho^{\sigma_l^+}$. We observe that we have zero-damped modes from the relation $s=\kappa\s$.
    \begin{remark*}
    Since the frequencies are generically no longer integers, the corresponding co-modes do not concentrate precisely the horizon as before. However, we conjecture that the co-modes in this case are close to the concentrated ones from before in the following sense: let $\theta$ be a co-mode for the equation with an inverse square potential. Then there exists $\theta_0$, a co-mode for the original equation that is concentrated on the horizon, such that $\theta-\theta_0$ is of order $V_0/(l+1/2)$ in operator norm.
    \end{remark*}
\subsection{The constant potential}\label{constant}
Before we consider a general class of potentials, let us consider the simple case when $W$ is a constant (i.e. we have a non-conformal mass for the Klein-Gordon equation). We can solve the equation using similar methods to the conformally coupled case and see that the solutions to $(L_s+W)u=0$ are:
\begin{align*}
    u(\mathbf{x})=\rho^lY_{lm}(\theta,\phi){}_2F_1\left[\frac{\s+l+1}{2}+\frac{1}{4}(1-\sqrt{1-4W}),\frac{\s+l+1}{2}+\frac{1}{4}(1+\sqrt{1-4W});\frac{3}{2}+l; \rho^2\right],
\end{align*}
where $\s=s/\kappa$. This allows us to read off the quasinormal spectrum as:
\begin{align*}
    s=-\kappa\left(2n+l+\frac{3}{2}\pm \frac{1}{2}\sqrt{1-4W}\right), \quad n\in \mathbb{N}_0,
\end{align*}
where $l$ is the angular appropriate angular momentum number. For $W\le 1/4$, we can see that we again have zero-damped modes clustering on the negative real axis and going to $0$ as $\kappa\rightarrow 0$. For $W>1/4$, the frequencies have non-zero imaginary part, but still exhibit the behaviour required of zero-damped mode frequencies. This is somewhat analogous to the situation in \cite{hod4}, where there exists a critical mass which determines whether the damping rate or the frequency is perturbed.
\subsection{Compactly supported potentials}
The fact that the co-modes are concentrated on the horizon for suitable values of $W$ (namely 0 and -2) implies that the behaviour of the potential near the horizon is all that matters when considering these zero-damped mode frequencies.
 \begin{prop}\label{potentialcomodes}
        Let $W \in C^{k}(\overline{B_1})$ such that
        \begin{align*}
            \partial^{\alpha}W|_{\partial B_1}=0
        \end{align*}
        for each multi-index $\alpha$ such that $|\alpha|\le k$. Take $n < k$. Then $-1, -2, \dots, -(n-1)$ are quasinormal frequencies of $L_{\s}+W: D^{n+1}\rightarrow H^n(B_1)$.
        \end{prop}
\begin{proof}
First we note that we have an analogue of \cref{Fredholm} for the equation with a $C^k$ potential using the results of \cite{cmw}. So it suffices to apply the co-modes from \cref{puredScomodes} to $(L_s+W)u$ and note that the equations are unchanged when the trace of these derivatives of $W$ is zero.
\end{proof}
In fact, this result also applies with multiplicity - this follows from finding the appropriate co-modes associated to higher total angular momentum spherical harmonics. We see in particular that if $W\in C^{\infty}_{c}(B_1)$ that $-\mathbb{N}$ is contained in the quasinormal spectrum of $L_s+W$. It is important to note at this point that this may not be the whole spectrum - additional frequencies could arise in $\realpart \s<-1/2$ due to the potential. Since $s=\kappa \s$, we have established the existence of zero-damped modes in this case.
    \subsection{Rouch\'e's theorem for analytic families of Fredholm operators}\label{roucheresult}
    We have established the existence of zero-damped modes for some very particular classes of potentials, however we wish to tackle more generic ones. To do so, we need some continuity result on the effect of perturbations of meromorphic families of operators on their poles. Rouch\'e's theorem is a classical result in complex analysis which achieves this type of control for zeroes of holomorphic functions, and Gohberg and Sigal established an analogue of this for meromorphic families of Fredholm operators in \cite{Gohberg_1971}. An alternate proof of their generalisation of the argument principle to holomorphic families of Fredholm operators is presented in \cite{logres} along with a discussion on which classes of Banach algebras such a result holds. We shall use the notation of \cite{zworski} in the statements of the theorems we quote (though we shall use the stronger version of \cref{rouche} that appears in \cite{Gohberg_1971}).
    \begin{defn}
    Suppose $X, Y$ are Banach spaces and $\Omega$ is a connected, open subset of $\mathbb{C}$. We say $A: \Omega \rightarrow \mathcal{L}(X,Y)$ is a meromorphic family of Fredholm operators in $\Omega$ if for any $z_0\in\Omega$, there exist finite rank operators $\{A_i\}_{k=1}^K$ and a neighbourhood of $z_0$ such that
    \begin{align*}
        A_0(z):=A(z)-\sum_{k=1}^K\frac{A_k}{(z-z_0)^k}
    \end{align*}
    is a holomorphic family of Fredholm operators in this neighbourhood.
    \end{defn}
    \begin{thm}\label{c.10}
		Suppose $X$ is a Banach space and $A(s)$ is a meromorphic family of Fredholm operators $X\rightarrow X$ in $\Omega$ where $\Omega$ is a connected, open subset of $\mathbb{C}$. Given $s_0\in\Omega$, there exists some neighbourhood of $s_0$ such that
		\begin{align*}
		    A(s)=A_0(s)+\sum_{k=1}^K\frac{A_k}{(s-s_0)^k},
		\end{align*}
		where $A_0(s)$ is holomorphic. If $A_0(s_0)$ has index zero, then there exist families of operators $U_1, U_2$ which are holomorphic and invertible in a neighbourhood of $s_0$ and operators $\{P_m\}_{m=0}^M$ such that 
		\begin{align*}
		    A(s)=U_1(s)\left(P_0+\sum_{m=1}^M(s-s_0)^{k_m}P_m\right)U_2(s)
		\end{align*}
		in a neighbourhood of $s_0$. The $k_m$ are non-zero integers and the $P_m$ obey
		\begin{align*}
		    P_mP_n=\delta_{mn}P_m,
		\end{align*}
		with $\{P_m\}_{m=1}^M$ being of rank 1 and $I-P_0$ being of finite rank. The inverse $A(s)^{-1}$ exists as a meromorphic family of operators near $s_0$ if and only if $\sum_{m=0}^MP_m=I$, in which case it takes the form:
		\begin{align*}
		    A(s)^{-1}=U_2(s)^{-1}\left(P_0+\sum_{m=1}^M(s-z_0)^{-k_m}P_m\right)U_1(s)^{-1}.
		\end{align*}
		\end{thm}
		\begin{proof}
		See Theorem C.10 in \cite{zworski} or Theorem 3.1 in \cite{Gohberg_1971}.
		\end{proof}
		\noindent This allows us to define the notion of null multiplicity:
		\begin{defn}
		The \textit{null multiplicity} at $s_0$ of a meromorphic family of operators $A(s)$ is:
		\begin{align*}
		    N_{s_0}(A):=\begin{cases}
		    \sum_{k_m>0}k_m & M=\rank(I-P_0)\\
		    \infty & M<\rank(I-P_0)
		    \end{cases}.
		\end{align*}
		When $N_{s_0}(A)<\infty$, $A(s)^{-1}$ is meromorphic and we can compute its null multiplicity:
		\begin{align*}
		    N_{s_0}(A^{-1})=-\sum_{k_m<0}k_m.
		\end{align*}
		\end{defn}
		\noindent Now we can state a generalisation of the argument principle for these families of operators:
		\begin{thm}\label{argumentprinciple}
		Let $H$ be a Hilbert space and suppose $A(s)$ and $A(s)^{-1}$ are meromorphic families of operators $H\rightarrow H$ in $\Omega$. Let
		\begin{align*}
		    \Pi_{s_0}=\frac{1}{2\pi i}\oint_{\Gamma_{s_0,\delta}}\partial_sA(s)A(s)^{-1}ds,
		\end{align*}
		where $\Gamma_{s_0,\delta}$ is a positively oriented circle of radius $\delta$ containing $s_0$ and no other pole of the integrand. Then $\Pi_{s_0}$ has finite rank and
		\begin{align*}
		    \tr\Pi_{s_0}=N_{s_0}(A)-N_{s_0}(A^{-1}).
		\end{align*}
		\end{thm}
		\begin{proof}
		See Theorem C.11 in \cite{zworski} or Theorem 2.1 in \cite{Gohberg_1971}.
		\end{proof}
		\begin{thm}\label{rouche}
		Suppose $A(s)$ and $B(s)$ for $s\in\Omega$ are meromorphic families of Fredholm operators as in \cref{argumentprinciple}. Suppose further that $U\Subset\Omega$ is a simply connected open subset with $C^1$ boundary $\partial U$ such that $A$ and $B$ are invertible on $\partial U$ and
		\begin{align*}
		    \norm{A(s)^{-1}(A(s)-B(s))}_{H\rightarrow H} <1, \quad s\in \partial U.
		\end{align*}
		Then
		\begin{align*}
		    \frac{1}{2\pi i}\tr \oint_{\partial U}\partial_sA(s)A(s)^{-1}ds=\frac{1}{2\pi i}\tr \oint_{\partial U}\partial_sB(s)B(s)^{-1}ds.
		\end{align*}
		\end{thm}
		\begin{proof}
		See Theorem 2.2 in \cite{Gohberg_1971}.
		\end{proof}
		\subsection{More general potentials}\label{polesdiscussion}
	The first important thing to note is that the results in the previous subsection require a family of operators which maps some Hilbert space to itself, so we need to modify the operators we are considering to apply them to our problem. We set
		\begin{align*}
		    A(s)&=L_sL_0^{-1}=I_{H^k}+2sPL_0^{-1}+s^2L_0^{-1},\\
		    B(s)&=A(s)+\epsilon WL_0^{-1},
		\end{align*}
		which are holomorphic families of Fredholm operators $H^k(B_1)\rightarrow H^k(B_1)$ for $W\in C^k(\overline{B_1})$, since $L_0: D^{k+1}(L_0)\rightarrow H^k(B_1)$ is invertible. We first note that by decomposing into spherical harmonics, we can find the eigenvalues and eigenvectors of $L_s$ in $D^{k+1}(L_s)$: this is equivalent to solving the equation for a different mass. We see that the eigenvectors (i.e. solutions of $L_su=\lambda u$) are
		\begin{align*}
		    u_{n,l,m}(\mathbf{x},s)=\rho^lY_{lm}(\theta,\phi){}_1F_2\left[-n,s+l+\frac{3}{2}+n;\frac{3}{2}+l;\rho^2\right]
		\end{align*}
		for $n\in\mathbb{N}_0$ with corresponding eigenvalues
		\begin{align*}
		    \lambda_{n,l,m}(s)=(s+l+1+2n)(s+l+2+2n).
		\end{align*}
		This means that for the factorisation of $A(s)$ given by \cref{c.10}, all the positive $k_m$ are 1 and since $A(s)$ is holomorphic in the region of interest, there are no negative $k_m$. Hence we have
		\begin{align*}
		    \frac{1}{2\pi i}\tr\oint_{\Gamma_{-n,1/2}}2(P+z)L_z^{-1}dz=n^2
		\end{align*}
		Noting that we can write
	    \begin{align*}
	        L_s^{-1}&=\frac{A_{-1}}{s+n}+A_0(s),
	    \end{align*}
	    we see that
	    \begin{align*}
	        A(s)^{-1}(A(s)-B(s))=\frac{\epsilon L_0A_{-1}WL_0^{-1}}{s+n}+\epsilon L_0A_0(s)WL_0^{-1}.
	    \end{align*}
	    We define the constants
	    \begin{align*}
	        C_{k,n}=\sup_{s\in \overline{D(-n,1/2)}}\left\{\norm{L_0A_0(s)WL_0^{-1}}_{H^k\rightarrow H^k}\right\},
	    \end{align*}
	    where $\overline{D(-n,1/2)}$ is the closed disc of radius $1/2$ around $-n$. We can use \cref{rouche} to obtain the following result.
	    \begin{prop}\label{freqcontrol}
	            Suppose $W\in C^k(\overline{B_1})$, pick $0<\delta<1/2$ and $n\in\mathbb{N}$ such that $n<k$. Let
	            \begin{align*}
	                \tilde{C}_{k,n}=\min \left\{\frac{1}{2\norm{L_0A_{-1}WL_0^{-1}}_{H^k\rightarrow H^k}},\frac{1}{C_{k,n}}\right\}.
	            \end{align*}
	            Then for $0<\epsilon<\tilde{C}_{k,n}\delta$, we have
	            \begin{align*}
	                \frac{1}{2\pi i}\tr\oint_{\Gamma_{-n,\delta}}2(P+z)(L_z+\epsilon W)^{-1}dz=\frac{1}{2\pi i}\tr\oint_{\Gamma_{-n,\delta}}2(P+z)L_z^{-1}dz.
	            \end{align*}
	    \end{prop}
	    \begin{proof}
	            We see that for $\epsilon < \tilde{C}_{k,n}\delta$, we have on a circle of radius $\delta$ around $-n$,
	            \begin{align*}
	                \norm{A(s)^{-1}(A(s)-B(s))}_{H^k\rightarrow H^k}&\le \frac{\epsilon}{\delta}\norm{L_0A_{-1}WL_0^{-1}}_{H^k\rightarrow H^k}+\epsilon C_{k,n}\\
	                &<\tilde{C}_{k,n}\norm{L_0A_{-1}WL_0^{-1}}_{H^k\rightarrow H^k}+\tilde{C}_{k,n}C_{k,n}\delta\\
	                &\le \frac{1}{2}+\delta <1.
	            \end{align*}
	            Thus the conditions in \cref{rouche} are met and we have
	            \begin{align*}
	                \frac{1}{2\pi i}\tr\oint_{\Gamma_{-n,\delta}}2(P+z)(L_z+\epsilon W)^{-1}dz=\frac{1}{2\pi i}\tr\oint_{\Gamma_{-n,\delta}}2(P+z)L_z^{-1}dz=n^2.
	            \end{align*}
	            \end{proof}
	   The upshot of this result is that we can control the distance, $\delta$, which a quasinormal frequency arising from a zero-damped mode can move from its unperturbed value provided the perturbation is of size $O(\delta)$.
	   \begin{thm}
	   Let $W\in C^{\infty}(\overline{B_1};\mathbb{R})$ and fix $n\in\mathbb{N}$ and $0<\delta<1/2$. We assume further that $\norm{W}_{C^{k}}=1$. Then there exists $M_k>0$ such that for $\epsilon<M_k\delta$, there exist $n^2$ quasinormal frequencies (counted with multiplicity) of $L_s+\epsilon W$ inside $D(-n,\delta)$ for $n\in\{1,2,\dots, k\}$.
	   \end{thm}
	   \begin{proof}
	   This is simply an application of \cref{freqcontrol} above.
	   \end{proof}
		\subsection{Computing the perturbed spectrum}
		The above result tells us the maximum size a potential can be given how much control we want on the frequencies. It does not give an explicit expression for the perturbed frequencies and how they depend on the potential we add to the equation. In the case of a spherically symmetric potential, we can achieve this in principle: we can obtain a series expansion for the perturbed frequencies.\\\\
		The $SO(3)$ symmetry of the system means that at frequency $-n$, there are $n^2$ modes corresponding to different spherical harmonics. Each of these will be affected by $W$ in a different way in general and thus it is difficult to get precise information on how they are perturbed in this situation. When $W$ is spherically symmetric, multiplication by $W$ commutes with projections onto any angular momentum sector, so we can separate the equation and lift this degeneracy. If $W$ is a non-trivial superposition of angular modes, this doesn't work since there is mixing between angular momentum sectors and the equations are coupled. At this stage, the poles are still simple and the residues are now rank one operators, which enables us to use \cref{argumentprinciple} to get an expression for perturbed frequencies.
		\\\\
		First we define the projections. For $u\in C^{\infty}(\overline{B_1})$, let us define \begin{align*}
		    (\Pi_{lm}u)(r,\theta,\phi)=Y_{lm}(\theta,\phi)\int_{S^2}u(r,\theta',\phi')Y_{lm}(\theta',\phi')\sin^2\theta ' d\theta' d\phi '.
		\end{align*}
		One can show that this is in fact bounded on $C^{\infty}(\overline{B_1})$ with respect to the $H^k(B_1)$ norm for each $k\ge 1$ and so extends to a continuous linear map $H^k(B_1)\rightarrow H^k(B_1)$.\\\\
		We assume that $W\in C^{k}(\overline{B_1})$ is spherically symmetric and apply the above methods on $\Pi_{lm}(L_s+W)$ where $\Pi_{lm}$ is the projection onto the angular momentum sector associated with the spherical harmonic $Y_{lm}(\theta,\phi)$. We shall omit explicitly writing the projection $\Pi_{lm}$ when considering operators like $(\Pi_{lm}L_s)^{-1}$ in the following calculations to reduce clutter. Then for $n< k$, consider the functional:
		\begin{align}
		    s_n[W]:=&\frac{1}{2\pi i}\tr \oint_{\Gamma_{-n,\delta}} z\partial_z( L_z+W)\left(L_z+W\right)^{-1}dz\\=&\frac{1}{2\pi i}\tr \oint_{\Gamma_{-n,\delta}} 2z(P+z)\left(L_z+W\right)^{-1}dz,
		\end{align}
		where $\delta<1$. We note that
		\begin{align*}
		    s_n[0]&=\frac{1}{2\pi i}\tr \oint_{\Gamma_{-n,\delta}} z\partial_z( L_z)L_z^{-1}dz\\
		    &=\frac{1}{2\pi i}\tr \oint_{\Gamma_{-n,\delta}} (z+n)\partial_z( L_z)L_z^{-1}dz-\frac{n}{2\pi i}\tr \oint_{\Gamma_{-n,\delta}}\partial_z( L_z)L_z^{-1}dz
		\end{align*}
		and since $L_z^{-1}$ has a simple pole at $z=-n$, the integrand in the first term is holomorphic and hence the integral evaluates to 0. Since the residue is a one-dimensional projection, it follows that $s_n[0]=-n$. Now, take $0<\epsilon< 1$ sufficiently small so that
		\begin{align*}
		    \frac{1}{2\pi i}\tr \oint_{\Gamma_{-n,\delta}}2(P+z)L_z^{-1}dz=\frac{1}{2\pi i}\tr \oint_{\Gamma_{-n,\delta}}2(P+z)(L_z+\epsilon W)^{-1}dz.
		\end{align*}
		This ensures that we are enclosing just one pole of $(L_s+\epsilon W)^{-1}$ and that it is also simple. We observe that
		\begin{align*}
		     s_n[\epsilon W]+n&=\frac{1}{2\pi i}\tr \oint_{\Gamma_{-n,\delta}} 2z(P+z)\left(\left(L_z+\epsilon W\right)^{-1}-L_z^{-1}\right)dz\\
		     &=\frac{1}{2\pi i}\tr \oint_{\Gamma_{-n,\delta}} 2z(P+z)L_z^{-1}\left(\left(I+\epsilon WL_z^{-1}\right)^{-1}-I\right)dz.
		\end{align*}
		Since $L_z$ is invertible for $z\in\Gamma_{-n,\delta}$, $WL_z^{-1}$ is bounded in the operator norm topology along the contour and hence its operator norm will take a maximum value on $\Gamma_{-n,\delta}$ by continuity. So if we take $\epsilon$ sufficiently small, we have $\norm{\epsilon W L_z^{-1}}<1$ along the contour and we may use a von Neumann series to expand it:
		\begin{align*}
     s_n[\epsilon W]+n=\frac{1}{2\pi i}\tr \oint_{\Gamma_{-n,\delta}} 2z(P+z)L_z^{-1}\sum_{m=1}^{\infty}(-\epsilon)^m(WL_z^{-1})^mdz.
\end{align*}
By uniform convergence in the operator norm topology of the sum, we may exchange the sum and the integral and consider for each $m$ the following:
\begin{align*}
    \frac{1}{2\pi i}\oint_{\Gamma_{-n,\delta}} 2z(P+z)L_z^{-1}(WL_z^{-1})^mdz.
\end{align*}
We write
\begin{align*}
    L_z^{-1}=\frac{A_{-1}}{z+n}+A_0(z)
\end{align*}
and consider its poles to see that the radius of convergence of the Taylor series of $A_0(z)$ about $-n$ is 1, so we can write in this disc
\begin{align*}
    A_0(z)=A_0(-n)+\sum_{m=1}^{\infty}\frac{A_0^{(m)}(-n)}{m!}(z+n)^m,
\end{align*}
where $A_0^{(m)}(z)$ is the $m$th derivative of $A_0$. Let us write $A_0=A_0(-n)$ and, for $m\in\mathbb{N}$,
\begin{align*}
    A_m:=\frac{A_0^{(m)}(-n)}{m!}.
\end{align*}
This gives us the expansion
\begin{align*}
    L_z^{-1}(WL_z^{-1})^m=\sum_{i_0=-1}^{\infty}\dots\sum_{i_m=-1}^{\infty}A_{i_0}(WA_{i_1})\dots (WA_{i_m})(z+n)^{i_0+\dots +i_m}
\end{align*}
and hence (noting that $2z(P+z)=2(z+n)^2+2(P-2n)(z+n)-2n(P-n)$), we see that
\begin{align*}
    2z(P+z)L_z^{-1}(WL_z^{-1})^m=&2\sum_{i_0=-1}^{\infty}\dots\sum_{i_m=-1}^{\infty}A_{i_0}(WA_{i_1})\dots (WA_{i_m})(z+n)^{i_0+\dots +i_m+2}\\
    &+2(P-2n)\sum_{i_0=-1}^{\infty}\dots\sum_{i_m=-1}^{\infty}A_{i_0}(WA_{i_1})\dots (WA_{i_m})(z+n)^{i_0+\dots +i_m+1}\\
    &-2n(P-n)\sum_{i_0=-1}^{\infty}\dots\sum_{i_m=-1}^{\infty}A_{i_0}(WA_{i_1})\dots (WA_{i_m})(z+n)^{i_0+\dots +i_m}.
\end{align*}
Using the residue theorem, we see that
\begin{align*}
    \frac{1}{2\pi i}\oint_{\Gamma_{-n,\delta}} 2z(P+z)L_z^{-1}(WL_z^{-1})^mdz=&2\sum_{i_0+\dots +i_m=-3}A_{i_0}(WA_{i_1})\dots (WA_{i_m})\\
    &+2(P-2n)\sum_{i_0+\dots +i_m=-2}A_{i_0}(WA_{i_1})\dots (WA_{i_m})\\
    &-2n(P-n)\sum_{i_0+\dots +i_m=-1}A_{i_0}(WA_{i_1})\dots (WA_{i_m}).
\end{align*}
Since $A_0(s)$ is holomorphic at $-n$ with radius of convergence 1, $\norm{A_m}_{H^k\rightarrow H^k}<1$. Let $B_k=\max\{1,\norm{A_{-1}}_{H^k\rightarrow H^k}\}$
\begin{lemma}\label{convergence}
For $\epsilon<1/(16B_k\norm{W}_{C^k})$, 
\begin{align*}
    \sum_{m=1}^{\infty}(-\epsilon)^m\oint_{\Gamma_{-n,\delta}} 2z(P+z)L_z^{-1}(WL_z^{-1})^mdz
\end{align*}
converges in the trace norm topology.
\end{lemma}
\begin{proof}
We note that since, $i_0+\dots +i_m<0$, each of the $A_{i_0}(WA_{i_1})\dots (WA_{i_m})$ is a rank 1 operator and hence the integral gives a finite rank operator with rank bounded by the number of $(i_0,i_1,\dots, i_m)$ which obey the conditions given. By setting $j_p=i_p+1$ for $p=1,2,\dots m$, we see that each of the conditions are equivalent to one of the following:
\begin{align*}
    j_0+j_1+\dots + j_m&=m,\\
    j_0+j_1+\dots + j_m&=m-1,\\
    j_0+j_1+\dots + j_m&=m-2,
\end{align*}
for $j_p$ non-negative integers. By standard arguments, the rank of the operator is at most
\begin{align*}
    \binom{2m}{m}+\binom{2m-1}{m}+\binom{2m-2}{m}<3\cdot 2^{2m}.
\end{align*}
We can bound the trace norm $\norm{\cdot}_{1}$ of a finite rank operator on a Hilbert space by its operator norm and its rank, so we have:
\begin{align*}
\left\lVert{\frac{1}{2\pi i}\oint_{\Gamma_{-n,\delta}} 2z(P+z)L_z^{-1}(WL_z^{-1})^mdz}\right\rVert_{1}<3\cdot 4^{m}\left\lVert{\frac{1}{2\pi i}\oint_{\Gamma_{-n,\delta}} 2z(P+z)L_z^{-1}(WL_z^{-1})^mdz}\right\rVert_{H^k\rightarrow H^k}.
\end{align*}
We also have the following bound:
\begin{align*}
    \norm{A_{i_0}(WA_{i_1})\dots (WA_{i_m})}_{H^k\rightarrow H^k}\le B_k^{m+1}\norm{W}_{C^k}^m
\end{align*}
and hence
\begin{align*}
    \left\lVert{\frac{1}{2\pi i}\oint_{\Gamma_{-n,\delta}} 2z(P+z)L_z^{-1}(WL_z^{-1})^mdz}\right\rVert_{1}<18 M 16^mB_k^{m+1}\norm{W}_{C^k}^m,
\end{align*}
where $M$ depends on $n$ and the operator norm of $P$. This gives the result. 
\end{proof}
Thus we have the following result:
\begin{thm}\label{seriesexpansion}
Let $W\in C^{\infty}(\overline{B_1};\mathbb{R})$ be a spherically symmetric potential. Then for $\epsilon$ sufficiently small, there exist quasinormal frequencies $s_n$ of $L_s+\epsilon W$ and $S_m\in \mathbb{C}$ such that
\begin{align*}
    s_n=-n+\sum_{m=1}^{\infty}S_m\epsilon^m.
\end{align*}
Furthermore, we can compute the $S_m$ explicitly:
\begin{align*}
    S_m=\frac{(-1)^m}{2\pi i}\tr \oint_{\Gamma_{-n,\delta}} 2z(P+z)L_z^{-1}(WL_z^{-1})^mdz.
\end{align*}
\end{thm}
\begin{proof}
We begin by projecting into any given angular momentum sector $(l,m)$ as above and performing the same steps to obtain a von Neumann series within a trace. Since the von Neumann series given above converges absolutely in the trace norm topology for small enough $\epsilon$ by \cref{convergence}, we can use continuity of trace to exchange the summation and trace operations to get
\begin{align*}
     s_n[\epsilon W]+n&=\frac{1}{2\pi i}\sum_{m=1}^{\infty}(-1)^m\epsilon^m\tr \oint_{\Gamma_{-n,\delta}} 2z(P+z)L_z^{-1}(WL_z^{-1})^mdz\\
     &=\sum_{m=1}^{\infty}S_m\epsilon^m,
\end{align*}
where
\begin{align*}
    S_m=\frac{(-1)^m}{2\pi i}\tr \oint_{\Gamma_{-n,\delta}} 2z(P+z)L_z^{-1}(WL_z^{-1})^mdz.
\end{align*}
\end{proof}
The first order change to the frequency, $\delta s_n$, can be expressed as:
\begin{align*}
    \delta s_n=2\tr \left[(P-2n)A_{-1}WA_{-1}-n(P-n)(A_0WA_{-1}+A_{-1}WA_0)\right].
\end{align*}
We can write $A_{-1}=u_n\theta_n$ where $u_n$ is the quasinormal mode corresponding to this angular momentum sector and frequency normalized so that $\norm{u_n}_{H^k}=1$. Note that due to \cref{argumentprinciple}, the operator $-2n(P-n)A_{-1}$ is a projection, so the normalization of $\theta_n$ is fixed i.e. $-2n\theta_n(Pu_n)+2n^2\theta_n(u_n)=1$. This gives us
\begin{align*}
    \delta s_n =2\theta_n(Wu_n)\theta_n((P-2n)u_n)-2n\theta_n(PA_0Wu_n+WA_0Pu_n)+2n^2\theta_n(A_0Wu_n+WA_0u_n).
\end{align*}
In general, this is quite difficult to compute explicitly since it involves finding the holomorphic part of the inverse operator in a neighbourhood of the pole. This calculation is manageable for the simple case of $\delta s_1$ for a constant potential $W$, however it is lengthy and not crucial to the main arguments of this paper. As such, we shall omit it and simply state that the result is
\begin{align*}
    \delta s_1=-1.
\end{align*}
This agrees with our exact calculation in \cref{constant}: for $|W|<1/4$, the corresponding frequency is:
\begin{align*}
    s&=-\frac{3}{2}+ \frac{1}{2}\sqrt{1-4 W}\\
	        &=-\frac{3}{2}+\frac{1}{2}\left(1-\frac{4 W}{2}+O(W^2)\right)\\
	        &=-1-W+O(W^2).
\end{align*}
\subsection{The extremal limit}
Let us now take $V\in C^{\infty}(\mathbb{R}^3;\mathbb{R})$ which obeys the conditions
\begin{align*}
    |\mathbf{x}|^{|\alpha|+2}\partial^{\alpha}V(\mathbf{x})\rightarrow 0\quad \text{as}\quad |\mathbf{x}|\rightarrow \infty
\end{align*}
for all multi-indices $\alpha$ and consider the equation
\begin{align*}
    -\Box_g\psi + 2\kappa^2\psi +V\psi=0.
\end{align*}
We can perform the same manipulations to obtain \eqref{rescaledequation}. Let $\chi\in C^{\infty}(\overline{B_1})$ be such that $\chi \equiv 0$ on $B_{1/3}$ and $\chi \equiv 1$ on $B_1\setminus B_{2/3}$ and note that our problem is now considering the invertibility of
\begin{align*}
    L_{s/\kappa}+(1-\chi)W+\chi W.
\end{align*}
Since $1-\chi$ vanishes to all orders on the boundary, we see that $-\kappa \mathbb{N}$ is contained in the quasinormal spectrum of $L_{s/\kappa}+(1-\chi)W$ from previous results. Again, it is important to note that there may be additional frequencies that result due to the addition of this potential, which will be important when we try to apply \cref{rouche} to this operator while treating $\chi W$ as the perturbation. The following lemma establishes that the perturbation due to $\chi W$ gets smaller as $\kappa \rightarrow 0$ because of the decay of the potential.
\begin{lemma}\label{extrlmt}
Let $\chi W$ be defined as above. Then given $\epsilon>0$, there exists a constant $R$ depending on $\epsilon$ and $W$ and a constant $C$ depending only on $k$ and $\chi$ such that for all $0<\kappa<1/(3R)$,
\begin{align*}
    \norm{\chi W}_{C^k}<C\epsilon.
\end{align*}
\end{lemma}
\begin{proof}
We note that the condition on $V$ means that given $\epsilon>0$, there exists $R>1/\epsilon$ such that for $|\mathbf{x}|>R$,
\begin{align*}
    |\mathbf{x}|^{|\alpha|+2}|\partial^{\alpha} V(\mathbf{x})|<\epsilon
\end{align*}
for each multi-index $\alpha$ such that $|\alpha|\le k$. Thus if $0< \kappa< 1/(3R)$, we have for $|\mathbf{x}|>1/3$,
\begin{align*}
    |\partial^{\alpha}W(\mathbf{x})|&=\frac{1}{\kappa^{|\alpha|+2}}|(\partial^{\alpha}V)(\mathbf{x}/\kappa)|\\
    &=\frac{1}{|\mathbf{x}|^{|\alpha|+2}}\left\lvert\frac{|\mathbf{x}|^{|\alpha|+2}}{\kappa^{|\alpha|+2}}(\partial^{\alpha}V)(\mathbf{x}/\kappa)\right\lvert\\
    &<3^{|\alpha|+2}\epsilon\\
    &<3^{k+2}\epsilon.
\end{align*}
This gives us
\begin{align*}
    |\partial^{\alpha}(\chi W)(\mathbf{x})|&\le \sum_{\beta\le\alpha}\binom{\alpha}{\beta}|\partial^{\beta}W(\mathbf{x})||\partial^{\alpha-\beta}\chi(\mathbf{x})|\\
    & < \norm{\chi}_{C^{k}}\sum_{\beta\le\alpha}\binom{\alpha}{\beta}3^{|\beta|+2}\epsilon\\
    &\le 9\epsilon\norm{\chi}_{C^{k}}3^k\sum_{\beta\le\alpha}\binom{\alpha}{\beta}\\
    & =9\cdot 6^k \epsilon \norm{\chi}_{C^{k}},
\end{align*}
which gives the estimate.
\end{proof}
With this in mind, we define
\begin{align*}
    \widehat{\chi W}:=\frac{\chi W}{\norm{\chi W}_{C^k}}
\end{align*}
and note the operator norm of this as a map $H^k(B_1)\rightarrow H^{k}(B_1)$ is independent of $\kappa$. We write
\begin{align*}
    A^W(s)&=(L_{s}+(1-\chi)W)L_0^{-1},\\
    B^W(s)&=A^{W}(s)+\chi W L_0^{-1},
\end{align*}
and note that as before we can decompose
\begin{align*}
    (L_{s}+(1-\chi)W)^{-1}=\frac{A^{W}_{-1}}{s+n}+A_0^W(s),
\end{align*}
where $A_0^W(s)$ is bounded in an open neighbourhood of $-n$. The fact that this is a simple pole follows from noting that $L_s'L_s^{-1}$ has simple poles and that the kernel of $L_s'$ is trivial at each of these. Thus there exists $0<r<1/2$ such that $A_0(s)$ is bounded on the disc of radius $r$ around $-n$. Note that we have suppressed the dependence of $r$ on $k$ and $n$ - this will in general depend on how many extra frequencies are introduced by the potential. Let us define
\begin{align*}
	        C_{k,n}^W=\sup_{s\in \overline{D(-n,r)}}\left\{\norm{L_0A_0^W(s)\widehat{\chi W}L_0^{-1}}_{H^k\rightarrow H^k}\right\}
	    \end{align*}
and
\begin{align*}
    \tilde{C}_{k,n}^W=\min \left\{\frac{1}{2\norm{L_0A^W_{-1}\widehat{\chi W}L_0^{-1}}_{H^k\rightarrow H^k}},\frac{1}{C^W_{k,n}}\right\}.
\end{align*}
We can obtain analogous results to before.
\begin{prop}\label{mainresultprop}
For each $n\in \mathbb{N}$ and $0<\delta<r(n,k)$, there exists $R>0$ such that for all $0<\kappa<1/(3R)$, there exists at least one quasinormal frequency $s$ inside $D(-n\kappa, \delta \kappa)$ i.e.
\begin{align*}
    |s+n\kappa |<\delta\kappa.
\end{align*}
\end{prop}
\begin{proof}
This is analogous to the proofs of the results from before. We see again that on the circle of radius $\delta$ around $-n$ that
\begin{align*}
    \norm{A^W(\s)^{-1}(A^W(\s)-B^W(\s))}_{H^k\rightarrow H^k}&\le \frac{\norm{\chi W}_{C^k}}{\delta}\norm{L_0A^W_{-1}\widehat{\chi W}L_0^{-1}}_{H^k\rightarrow H^k}+\norm{\chi W}_{C^k} C^W_{k,n},
\end{align*}
where again $\s=s/\kappa$. We note that by \cref{extrlmt}, there exists $R$ such that for all $0<\kappa<1/(3R)$ $\norm{\chi W}_{C^k}<\tilde{C}^W_{k,n}\delta$. Thus we have
\begin{align*}
    \norm{A^W(\s)^{-1}(A^W(\s)-B^W(\s))}_{H^k\rightarrow H^k}&< \tilde{C}^W_{k,n}\norm{L_0A^W_{-1}\widehat{\chi W}L_0^{-1}}_{H^k\rightarrow H^k}+\tilde{C}^W_{k,n} C^W_{k,n}\delta\\
    &\le \frac{1}{2}+\delta<1,
\end{align*}
since $\delta<r\le 1/2$. The conditions in \cref{rouche} are met and so the number of quasinormal frequencies (counted with multiplicity) are unchanged. So there exists at least one $s$ such that
\begin{align*}
    \left\lvert \frac{s}{\kappa}+n\right\lvert&<\delta\\
    \Rightarrow |s+n\kappa|&<\delta\kappa <\frac{\kappa}{2}.
\end{align*}
\end{proof}
\begin{thm}\label{mainresultthm}
Let $V\in C^{\infty}(\mathbb{R}^3;\mathbb{R})$ be such that
\begin{align*}
    |\mathbf{x}|^{|\alpha|+2}\partial^\alpha V(\mathbf{x})\rightarrow 0 \quad \text{as} \quad |\mathbf{x}|\rightarrow \infty
\end{align*}
for each multi-index $\alpha$. Let $g$ be the metric on de Sitter as outlined in Section 2. Then the equation
\begin{align*}
    -\Box_{g}\psi+2\kappa^2\psi+V\psi=0
\end{align*}
exhibits the phenomenon of zero-damped quasinormal frequencies.
\end{thm}
\begin{proof}
This is simply an application of \cref{mainresultprop}, noting that for each $\kappa$ sufficiently small, there is a subset of the quasinormal spectrum $\{s_{n}\}_{n=1}^{\infty}$ such that
\begin{align*}
    |s_n+n\kappa|<\frac{\kappa}{2}.
\end{align*}
\end{proof}
This is sufficient to establish the existence of these frequencies, however we don't have much control on the size of the error - we just have $O(\kappa)$ circles where we may find them (see \cref{fig:qnf}).
\begin{figure}
    \centering
    \begin{subfigure}[b]{0.5\textwidth}
    \centering
    \begin{tikzpicture}
\node (A) at ( 0,3){$\Im s$};
\node (B) at (1,0){$\Re s$};
\node (C) at (0,-3){};
\draw[<-] (A) -- (C);
\draw[->] (-6.5,0) -- (B);
\draw (0,0) -- ({-2.5*sqrt(3)},2.5);
\draw (0,0) -- ({-2.5*sqrt(3)},-2.5);
\filldraw[color=red!60, fill=red!30, thick,fill opacity=0.5](-4,0) circle (2);
\filldraw[color=blue!60, fill=blue!30, thick,fill opacity=0.5](-2,0) circle (1);
\filldraw[color=black!60!green!60, fill=black!30!green!30, thick,fill opacity=0.5](-1,0) circle (0.5);
\filldraw[color=black!60!green](-1,0) circle (0.05);
\filldraw[color=blue](-2,0) circle (0.05);
\filldraw[color=red](-4,0) circle (0.05);
\end{tikzpicture}
\caption{$s=-\kappa$}
    \end{subfigure}%
    \begin{subfigure}[b]{0.5\textwidth}
    \centering
    \begin{tikzpicture}
\node (A) at ( 0,3){$\Im s$};
\node (B) at (1,0){$\Re s$};
\node (C) at (0,-3){};
\draw[<-] (A) -- (C);
\draw[->] (-6.5,0) -- (B);
\draw (0,0) -- (-6.5,{6.5/sqrt(15)});
\draw (0,0) -- (-6.5,{-6.5/sqrt(15)});
\filldraw[color=red!60, fill=red!30, thick,fill opacity=0.5](-4,0) circle (1);
\filldraw[color=blue!60, fill=blue!30, thick,fill opacity=0.5](-2,0) circle (0.5);
\filldraw[color=black!60!green!60, fill=black!30!green!30, thick,fill opacity=0.5](-1,0) circle (0.25);
\filldraw[color=black!60!green](-1,0) circle (0.05);
\filldraw[color=blue](-2,0) circle (0.05);
\filldraw[color=red](-4,0) circle (0.05);
\end{tikzpicture}
\caption{$s=-2\kappa$}
    \end{subfigure}
    \caption{The shaded regions correspond to the possible location of quasinormal frequencies for differing values of $\kappa$. We see that as $\kappa$ gets smaller, the frequency approaches 0 in a sector with angle $\arcsin(1/2n)$.}
    \label{fig:qnf}
\end{figure}
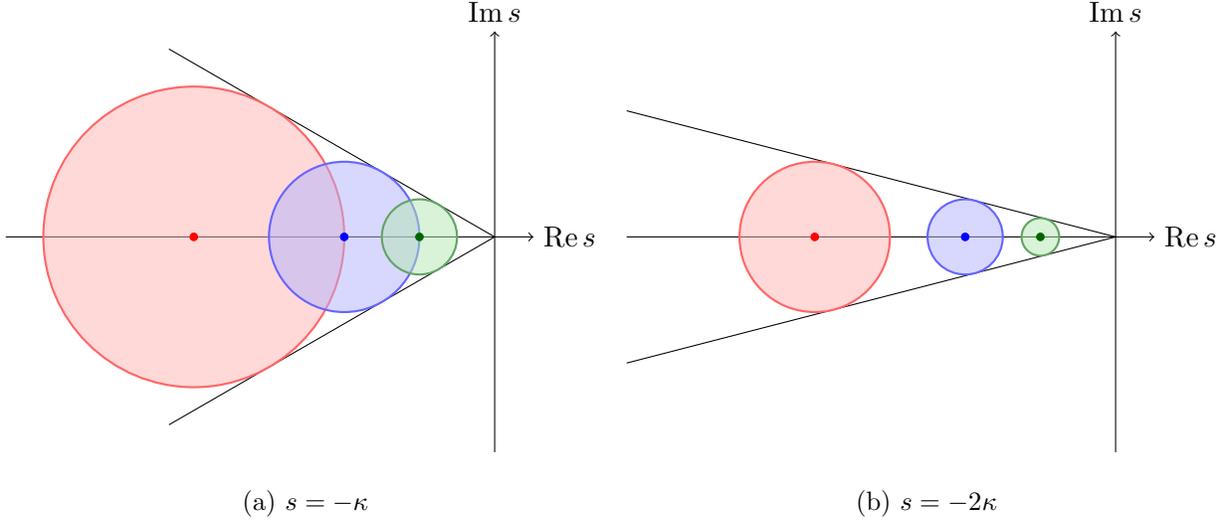
If we assume better decay of $V$, we can get a better idea of how far away this subset of the spectrum can stray from $-\kappa\mathbb{N}$. An example of such a condition would be
\begin{align*}
    |\mathbf{x}|^{|\alpha|+3}|\partial^\alpha V(\mathbf{x})|<M_1
\end{align*}
for some constant $M_1>0$ and for each multi-index $\alpha$. This would give us existence of quasinormal frequencies $s_n(\kappa)$ such that for $\kappa$ small enough,
\begin{align*}
    |s_n(\kappa)+\kappa n|<M_2\kappa^2
\end{align*}
for some $M_2>0$ independent of $\kappa$ and thus $s_n(\kappa)=-n\kappa + O(\kappa^2)$.
\section{Perturbations to the metric}\label{blackhole}
The arguments above are robust in the sense that they can be applied to more general spherically symmetric black hole spacetimes with a cosmological horizon. We again consider the conformally coupled Klein-Gordon equation but now with the metric
\begin{align*}
    g=-f_\Lambda(r)dt^2+\frac{dr^2}{f_{\Lambda}(r)}+r^2g_{S^2}
\end{align*}
where 
\begin{align*}
    f_{\Lambda}(r)=1+w_{\Lambda}(r)+\frac{\Lambda}{3}\alpha_\Lambda r-\frac{\Lambda}{3}r^2
\end{align*}
for some $\alpha_\Lambda\in\mathbb{R}$ bounded for $\Lambda\in[0,\Lambda_0]$ such that $\alpha_\Lambda\rightarrow \alpha_0$ as $\Lambda\rightarrow 0$ and with $w_{\Lambda}(r)\in C^{\infty}((0,\infty);\mathbb{R})$ for each $\Lambda\in [0,\Lambda_0]$ obeying the following additional conditions:
\begin{enumerate}
    \item for each $k\in\mathbb{N}_0, \Lambda\in [0,\Lambda_0]$, $r^k\partial_r^k w_{\Lambda}(r)\rightarrow 0$ as $r\rightarrow \infty$,
    \item $1+w_{0}(r)$ has at least one root and finitely many others. We further assume the largest root, $r_e(0)$, is simple and $w_0(r)<0$ for $r>r_e(0)$,
    \item there exists $0<t<r_e(0)$ such that for each $k\in\mathbb{N}_0$, there exists $\beta_k>0$ so the following holds
    \begin{align*}
        \sup_{r\ge t}\left|r^k\partial_r^kw_{\Lambda}(r)-r^k\partial_r^kw_0(r)\right|<\frac{\Lambda}{3}\beta_k.
    \end{align*}
   
\end{enumerate}
These conditions are sufficient by standard arguments (for sufficiently small $\Lambda$) to establish the following properties:
\begin{enumerate}
    \item $f_\Lambda(r)$ has a simple root at $r=r_c(\Lambda)$ (cosmological horizon) such that $r_c(\Lambda)\rightarrow \infty$ and $\Lambda r_c^2/3\rightarrow 1$ as $\Lambda\rightarrow 0$,
    \item $f_\Lambda(r)$ has a simple root at $r=r_e(\Lambda)$ (event horizon) such that $r_e(\Lambda)\rightarrow r_e(0)$ as $\Lambda\rightarrow 0$,
    \item $f_\Lambda(r)$ has no roots in the interval $(r_e(\Lambda),r_c(\Lambda))$ for $\Lambda$ sufficiently small,
    \item $1-f_\Lambda(r)>0$ for $r_e(\Lambda)<r<r_c(\Lambda)$ provided $\Lambda$ is sufficiently small and furthermore, there exists $\epsilon>0$ such that $1-f_\Lambda(r)>\epsilon$ for any $r_c(\Lambda)/3<r<r_c(\Lambda)$.
\end{enumerate}
Important examples of $w_{\Lambda}$ which obey these conditions are $w_{\Lambda}(r)=-2m/r$ (the Schwarzschild-de Sitter solution) and $w_\Lambda(r)=-2m/r+q^2/r^2$ (Reissner-Nordstr\"om-de Sitter). For an example of $w_{\Lambda}(r)$ which depends on $\Lambda$, see the next subsection. From now on, we shall suppress the dependence on $\Lambda$ of some quantities. The equation we shall consider is the conformal Klein-Gordon equation for this spacetime:
\begin{align*}
    \Box_g \psi -\frac{R}{6}\psi=0,
\end{align*}
where $R$ is the Ricci scalar of the spacetime. In this case, we see that
\begin{align*}
    R&=4\Lambda-\frac{2\Lambda\alpha_{\Lambda}}{r}-w_{\Lambda}''(r)-\frac{4w_{\Lambda}'(r)}{r}-\frac{2w_{\Lambda}(r)}{r^2}\\
    &=4\Lambda +\frac{\Lambda}{3}V^w(r).
\end{align*}
We use a coordinate transformation analogous to the pure de Sitter case to get a hyperboloidal slicing:
\begin{align*}
    \tau &= t - \int_{r_0}^r\frac{\sqrt{1-f_{\Lambda}(\xi)}}{f_{\Lambda}(\xi)}d\xi, \\
    \rho &= \frac{r}{r_c}.
\end{align*}
where $r_0$ is fixed so that it lies in $(r_e(\Lambda),r_c(\Lambda))$ for all $\Lambda\in[0,\Lambda_0]$. Note that the fourth condition given above is necessary to make this transformation well-defined. As before (since the horizons we consider have non-zero surface gravity) the results in \cite{cmw} allow us to define the quasinormal frequencies and similar reasoning to \cref{review} means it suffices to consider following equation on $\overline{B_1 \setminus B_{\rho_e}}$:
\begin{align}
\begin{split}
    -\frac{f_{\Lambda}(r_c\rho)}{r_c^2}\partial_\rho^2u&-\left(\frac{f_{\Lambda}'(r_c\rho)}{r_c}+\frac{2}{r_c^2\rho}f_{\Lambda}(r_c\rho)-\frac{2s}{r_c}\sqrt{1-f_{\Lambda}(r_c\rho)}\right)\partial_\rho u\\&+\left(s^2+\frac{2\Lambda}{3}+\frac{2s}{r_c\rho}\sqrt{1-f_{\Lambda}(r_c\rho)}-\frac{s}{2}\frac{f_{\Lambda}'(r_c\rho)}{\sqrt{1-f_{\Lambda}(r_c\rho)}}+\frac{\Lambda}{3}V^w(r_c\rho)\right)u-\frac{\slashed{\Delta}u}{r_c^2\rho^2}=F,
\end{split}
\end{align}
where $\rho_e=r_e/r_c$ and $F$ is a function constructed from the initial data. It will be useful to introduce the notation $U_r:=B_1\setminus B_r$ from this point onwards. Dividing through by $3/\Lambda$ and setting $\s=s\sqrt{3/\Lambda}$ we get the operator
\begin{align}
\begin{split}
    L_{\s}^{w}u:=&-\frac{3}{\Lambda r_c^2}f_{\Lambda}(r_c\rho)\partial_\rho^2u-\left(\frac{3}{\Lambda r_c^2 }\left(r_cf_{\Lambda}'(r_c\rho)+\frac{2f_{\Lambda}(r_c\rho)}{\rho}\right)-2\s\sqrt{\frac{3}{\Lambda r_c^2}\left(1-f_{\Lambda}(r_c\rho)\right)}\right)\partial_\rho u\\&+\left(\s^2+2+\frac{2\s}{\rho}\sqrt{\frac{3}{\Lambda r_c^2}\left(1-f_{\Lambda}(r_c\rho)\right)}-\frac{\s}{2}\sqrt{\frac{3}{\Lambda r_c^2}}\frac{r_cf_{\Lambda}'(r_c\rho)}{\sqrt{1-f_{\Lambda}(r_c\rho)}}+V^w(r_c\rho)\right)u\\&-\frac{3}{\Lambda r_c^2}\frac{\slashed{\Delta}u}{\rho^2}.
\end{split}
\end{align}
In order to use the results outlined in \cref{roucheresult}, we need to compare the above operator with $L_s$. These operators are, however, defined on different domains. To get around this, we define a smooth, spherically symmetric cut-off $\chi\in C^{\infty}(\overline{B_1})$ such that
\begin{align*}
    \chi(\mathbf{x})=\begin{cases}
    0 & |\mathbf{x}|<\frac{1}{3}\\
    1 & |\mathbf{x}|>\frac{2}{3}
    \end{cases}
\end{align*}
and note that for $\Lambda$ sufficiently small, this vanishes to all orders on the event horizon $\rho=\rho_e$. We now define the family of operators
\begin{align*}
    Q_s:= L_s^w\circ \chi - L_s\circ E_\chi,
\end{align*}
where $E_\chi$ multiplies a function by $\chi$ and extends it to $B_1$ by setting it to zero outside its domain of definition. Note also that we have omitted explicitly writing the map which restricts functions on $B_1$ to functions on $U_{\rho_e}$. This gives us the following decomposition of $L^w_s$:
\begin{align*}
    L^w_s=L^w_s\circ (1-\chi)+L_s\circ E_\chi +Q_s.
\end{align*}
The unbounded operator $L_s^w\circ (1-\chi)+L_s\circ E_\chi$ is a degenerate elliptic operator with the degeneracy occurring precisely at the event and cosmological horizons. Furthermore the surface gravity at these horizons is non-zero, so from the results of \cite{cmw}, this defines a family of Fredholm operators from its domain to $H^k(U_{\rho_e})$. We also have the following estimates:
\begin{lemma}
There exists $\Lambda_0>0$ such that for all $\Lambda<\Lambda_0$, there exists a constant $C>0$ depending only on $k$ such that
\begin{align*}
    \norm{L_0^w(1-\chi)u+L_0\circ E_\chi u)}_{H^k}&\le C\left(\norm{L_0^w u}_{H^k}+\norm{u}_{H^{k}}\right),\\
    \norm{L_0^w u}_{H^k}&\le C\left(\norm{L_0^w(1-\chi)u+L_0\circ E_\chi u)}_{H^k}+\norm{u}_{H^{k}}\right).
\end{align*}
\end{lemma}
\begin{proof}
Throughout this section the constant $C$ may change value from line to line to reduce clutter. The key point is that it depends only on $k$. We begin by proving the intermediary estimate 
\begin{align*}
    \norm{(\Delta-\partial_\rho^2)u}_{H^k}\le C\norm{u}_{D^{k+1}(L^w_s)}.
\end{align*}
We begin by observing that
\begin{align*}
    \norm{L^w_0(\chi u)}^2_{H^k}-\left(\frac{3}{\Lambda r_c^2}\right)^2&\norm{\left(\Delta-\partial_\rho^2\right)(\chi u)}_{H^k}^2\\
    &\ge -\frac{6}{\Lambda r_c^2}\left(L^w_0(\chi u)+\frac{3}{\Lambda r_c^2}\left(\Delta-\partial_\rho^2\right)(\chi u),\left(\Delta-\partial_\rho^2\right)(\chi u)\right)_{H^k},
\end{align*}
where $\Delta$ is the Laplacian on $U_{\rho_e}$. We aim to show that the right-hand side can be controlled by a suitable multiple of $\norm{u}_{D^{k+1}(L_s^w)}$. A useful observation from the theory developed in \cite{cmw} is the fact that there exists a constant $C$ depending only on $k$ such that $\norm{u}_{H^{k+1}}\le C\norm{L_0^wu}_{H^k}$, so it suffices to show that the terms we are interested in can be dominated by a multiple of $\norm{u}_{H^{k+1}}$. We switch to polar coordinates so $\Delta-\partial_\rho^2=(2/\rho)\partial_\rho +(1/\rho^2)\slashed{\Delta}$ and consider each of these separately. Since the operators we are interested in commute with angular derivatives, it suffices to prove this for radial ones only.
\begin{align*}
    \partial_\rho^k\left(\frac{2}{\rho}\partial_\rho(\chi u)\right)=\frac{2}{\rho}\partial_\rho^{k+1}(\chi u)+\sum_{j=1}^{k}p_j(\rho)\partial_\rho^j(\chi u)
\end{align*}
for some functions $p_j$ which are bounded on $\supp \chi$. We combine this with
\begin{align*}
    \partial_\rho^{k}\left(L^w_0(\chi u)+\frac{3}{\Lambda r_c^2}\left(\Delta-\partial_\rho^2\right)(\chi u)\right)=-\frac{3f_{\Lambda}}{\Lambda r_c^2}\partial_\rho^{k+2}(\chi u)&+\frac{3}{\Lambda r_c^2}\left(\frac{2}{\rho}(1-f_{\Lambda})-(k+1)\partial_\rho f_{\Lambda}\right)\partial_\rho^{k+1}(\chi u)\\
    &+\sum_{j=1}^kq_j(\rho)\partial_\rho^j(\chi u),
\end{align*}
where the $q_j$ are some functions constructed from $3f_{\Lambda}(r_c\rho)/(\Lambda r_c^2)$.  We observe that
\begin{align*}
    \frac{3}{\Lambda r_c^2}f_{\Lambda}(r_c\rho)-(1-\rho^2)=\frac{3}{\Lambda r_c^2}-1+\frac{3}{\Lambda r_c^2}w_\Lambda(r_c\rho)+\frac{\alpha_\Lambda\rho}{r_c}.
\end{align*}
Restricting to the support of $\chi$, we see that the right hand side and its derivatives converge uniformly to 0 as $\Lambda\rightarrow 0$ by observing that
\begin{align*}
    \left|\partial^k_\rho(w_{\Lambda}(r_c\rho))\right|&\le\frac{\left|r_c^k\rho^k(w_\Lambda^{(k)}(r_c\rho)-w_0^{(k)}(r_c\rho))\right|}{\rho^k}+\frac{\left| r_c^k\rho^kw_0^{(k)}(r_c\rho)\right|}{\rho^k}\\
    &<3^k\left(\frac{\Lambda}{3}\beta_k+|r_c^k\rho^kw_0^{(k)}(r_c\rho)|\right)
    \end{align*}
and that the quantity in the brackets above goes to 0 as $\Lambda\rightarrow 0$. Hence $3f_{\Lambda}(r_c\rho)/(\Lambda r_c^2)\rightarrow 1-\rho^2$ in $C^{l}(\supp \chi)$ as $\Lambda\rightarrow 0$ for each $l\in\mathbb{N}_0$. Thus, there exists $\Lambda_0>0$ such that for all $\Lambda<\Lambda_0$, we can control these functions by some constant. Taking the $L^2$-inner product of the expression above, we see that the only pairing that cannot clearly be controlled by a suitable multiple of $\norm{u}_{H^{k+1}}$ is
\begin{align*}
    -\frac{3}{\Lambda r_c^2}\int_{U_{\rho_e}}\frac{2f_{\Lambda}}{\rho} \partial_\rho^{k+1}(\chi \Bar{u})\partial_\rho^{k+2}(\chi u)dx&= -\frac{3}{\Lambda r_c^2}\int_{U_{\rho_e}}\frac{f_{\Lambda}}{\rho}\partial_\rho \left(\left|\partial_\rho^{k+1}(\chi u)\right|^2\right)dx\\
    &=\frac{3}{\Lambda r_c^2}\int_{U_{\rho_e}}\partial_\rho\left(\frac{f_{\Lambda}}{\rho}\right)\left|\partial_\rho^{k+1}(\chi u)\right|^2dx,
\end{align*}
where the boundary term vanishes due to $f_{\Lambda}$. Thus this term can also be made positive by adding a multiple of $\norm{u}_{H^{k+1}}$ where the constant is independent of $\Lambda$ for $\Lambda$ small enough. To deal with the other term, we write $\chi u=\rho^2 \cdot \chi u/\rho^2=\rho^2 U$ and consider
\begin{align*}
    \partial_\rho^{k}\left(L^w_0(\chi u)+\frac{3}{\Lambda r_c^2}\left(\Delta-\partial_\rho^2\right)(\chi u)\right)=\partial_\rho^{k}\left(L^w_0(\rho^2U)+\frac{3}{\Lambda r_c^2}\left(\Delta-\partial_\rho^2\right)(\rho^2U)\right).
\end{align*}
We only need to consider the two highest order derivative terms since after multiplying by $\slashed{\Delta}\partial_\rho^kU$ and integrating by parts on the sphere, we are left with terms of the form $t(\rho)\partial_\rho^l\slashed{\nabla}U\cdot \partial_\rho^k\slashed{\nabla}U$ for some function $t$ and $0\le l\le k+2$. These can be controlled by $\norm{u}_{H^{k+1}}$ for $l\le k$ so we only have to consider the terms with $l=k+1$ and $l=k+2$. These are
\begin{align}\label{highestorderterms}
    -\frac{3}{\Lambda r_c^2}\left(\rho^2f_{\Lambda}\partial_\rho^{k+2}U+2(k+2)\rho f_{\Lambda}\partial_\rho^{k+1}U+k\rho^2\partial_\rho f_{\Lambda}\partial_\rho^{k+1}U-\partial_\rho(\rho^2(1-f_{\Lambda}))\partial_\rho^{k+1}U\right).
\end{align}
We start with the highest order term:
\begin{align*}
    \left(\frac{3}{\Lambda r_c^2}\right)^2\int_{U_{\rho_e}}\rho^2f_{\Lambda}\partial_\rho^{k+2}\Bar{U}\slashed{\Delta}\partial_\rho^kUdx=&-\left(\frac{3}{\Lambda r_c^2}\right)^2\int_{U_{\rho_e}}\rho^2f_{\Lambda}\partial_\rho^{k+2}\slashed{\nabla}\Bar{U}\cdot\partial_\rho^k\slashed{\nabla}Udx\\
    =&\left(\frac{3}{\Lambda r_c^2}\right)^2\int_{U_{\rho_e}}\rho^2f_{\Lambda}|\partial_\rho^{k+1}\slashed{\nabla}U|^2dx\\
    &+\left(\frac{3}{\Lambda r_c^2}\right)^2\int_{U_{\rho_e}}(\rho^2\partial_\rho f_{\Lambda}+4\rho f_{\Lambda})\partial_\rho^{k+1}\slashed{\nabla}\Bar{U}\cdot\partial_\rho^k\slashed{\nabla}Udx\\
    \ge & \frac{1}{2}\left(\frac{3}{\Lambda r_c^2}\right)^2\int_{U_{\rho_e}}(\rho^2\partial_\rho f_{\Lambda}+4\rho f_{\Lambda})\partial_\rho|\partial_\rho^k\slashed{\nabla}U|^2dx.
\end{align*}
Since $f_{\Lambda}$ vanishes on the boundary, we can integrate by parts to deduce that the term involving $f_{\Lambda}$ can be dealt with using $\norm{u}_{H^{k+1}}$. Combining this with \eqref{highestorderterms}, we are left with
\begin{align*}
    \frac{1}{2}\left(\frac{3}{\Lambda r_c^2}\right)^2&\int_{U_{\rho_e}}(2\rho-k\rho^2\partial_\rho f_{\Lambda})\partial_\rho|\partial_\rho^k\slashed{\nabla}U|^2dx\\
    &\ge -\frac{1}{2}\left(\frac{3}{\Lambda r_c^2}\right)^2\int_{U_{\rho_e}}(6-4k\rho\partial_\rho f_{\Lambda}-k\rho^2\partial_\rho^2 f_{\Lambda})|\partial_\rho^k\slashed{\nabla}U|^2dx.
\end{align*}
The inequality follows from the fact that the cut-off allows us to ignore the inner boundary and that $\partial_\rho f_{\Lambda}<0$ at the outer boundary, so $2-k\partial_\rho f_{\Lambda}>0$. Thus this is also dominated by a suitable multiple of $\norm{u}_{H^{k+1}}$ and we have the following result: there exists $C,\Lambda_0>0$ depending only on $k$ such that for all $\Lambda<\Lambda_0$,
\begin{align*}
    \norm{\left(\Delta-\partial_\rho^2\right)(\chi u)}_{H^{k}}\le C\left(\norm{L^w_0(\chi u)}_{H^k}+\norm{u}_{H^{k+1}}\right)\le C\norm{u}_{D^{k+1}(L_s^w)}.
\end{align*}
We now use this to prove the lemma. From the product rule and the usual estimates,
\begin{align*}
    \norm{L_0^w(1-\chi)u+L_0\circ E_\chi u)}_{H^k}\le C\left(\norm{L_0^w u}_{H^k}+\norm{u}_{H^{k+1}}\right)+\norm{L_0\circ E_\chi u}_{H^k}.
\end{align*}
Since $f_{\Lambda}$ vanishes only at the horizons and both these zeroes are simple, $(1-\rho^2)/f_{\Lambda}(r_c\rho)$ is smooth and bounded on $\supp \chi$ and thus has a finite $C^{k}$ norm. Thus there exists a constant $C$ such that for all $\Lambda$ small enough,
\begin{align*}
    \norm{L_0\circ E_\chi u}_{H^k}&\le \norm{\frac{\Lambda r_c^2}{3}\frac{1-\rho^2}{f_{\Lambda}(r_c\rho)}\left(L_0^w(\chi u)+\frac{3}{\Lambda r_c^2}(\Delta-\partial_\rho^2)(\chi u)\right)}+C\norm{u}_{D^{k+1}(L^w_s)}\\
    &\le C\left(\norm{L_0^wu}_{H^k}+\norm{u}_{H^{k+1}}\right)\\
    &\le C\norm{u}_{D^{k+1}(L_s^w)}.
\end{align*}
Next we observe that
\begin{align*}
    \norm{L_0^wu}_{H^k}&=\norm{L_0^w(1-\chi)u+L_0^w(\chi u)}_{H^k}\\
    &\le \norm{L_0^w(1-\chi)u+L_0\circ E_\chi u}_{H^k} +\norm{L_0^w(\chi u)-L_0\circ E_\chi u}_{H^k}.
\end{align*}
Focussing on the second term, we have the following inequality
\begin{align*}
    \norm{L_0^w(\chi u)-L_0\circ E_\chi u}_{H^k}\le& \norm{\left((1-\rho^2)-3f_{\Lambda}/(\Lambda r_c^2)\right)\partial_\rho^2(\chi u)+\left(1-3/(\Lambda r_c^2)\right)(\Delta -\partial_\rho^2)(\chi u)}_{H^k}\\
    &+\norm{\left(\frac{3}{\Lambda r_c^2\rho^2}\partial_\rho(\rho^2(1-f_{\Lambda}(r_c\rho))-\frac{1}{\rho^2}\partial_\rho(\rho^2\cdot\rho^2)\right)\partial_\rho(\chi u)}_{H^k}+\norm{V^w\chi u}_{H^k}.
\end{align*}
We can deal with the term involving $V^w$ since
\begin{align*}
    V^w(r_c\rho)=-\frac{6\alpha_\Lambda}{r_c\rho}-\frac{3}{\Lambda r_c^2}\left(\frac{r_c^2\rho^2 w_{\Lambda}''(r_c\rho)}{\rho^2}+\frac{4r_c\rho w_{\Lambda}'(r_c\rho)}{\rho^2}+\frac{2w_{\Lambda}(r_c\rho)}{\rho^2}\right)
\end{align*}
and in $\supp\chi$ that $V^w(r_c\rho)$ and finitely many of its derivatives can be made arbitrarily small by similar reasoning to before. Thus we have both $V^w(r_c\rho)\rightarrow 0$ and $3f_{\Lambda}(r_c\rho)/(\Lambda r_c^2)\rightarrow 1-\rho^2$ in $C^{k+1}(\supp \chi)$ as $\Lambda\rightarrow 0$, so there exists $\Lambda_0>0$ such that for all $\Lambda<\Lambda_0$ we have all the results above and furthermore
\begin{align*}
    \norm{L_0^w(\chi u)-L_0\circ E_\chi u}_{H^k}\le \frac{1}{2}\norm{L_0^wu}.
\end{align*}
Therefore
\begin{align*}
    \frac{1}{2}\norm{L_0^w}\le \norm{L_0^w(1-\chi)u+L_0\circ E_\chi u}_{H^k}+C\norm{u}_{H^{k+1}}
\end{align*}
which yields the result.
\end{proof}
The lemma above establishes that $L_s^w\circ (1-\chi)+L_s\circ E_\chi$ is a holomorphic family of Fredholm operators $D^{k+1}(L_s^w)\rightarrow H^k(U_{\rho_e})$. Furthermore, for $\rho\in [1/3,1]$, the cut-off allows the distributions defined in \cref{comodesection} to satisfy the appropriate conditions to be co-modes of this family of operators at the usual frequencies i.e. $-\mathbb{N}$ is contained in the quasinormal spectrum. We shall treat $Q_s$ as a perturbation to this operator and apply the results of \cref{roucheresult}. Using similar notation, we set
\begin{align*}
    A^w(s)&=(L_s^w\circ (1-\chi)+L_s\circ E_\chi)(L_0^w)^{-1},\\
    B^w(s)&=A^w(s)+ Q_s(L_0^w)^{-1},
\end{align*}
noting that $L_0^w$ is invertible since its kernel is trivial and it is Fredholm of index 0. By \cref{rouche}, the number of frequencies contained within a contour (counted with multiplicity) is the same when
\begin{align*}
    \norm{L_0^w\left(L_s^w\circ (1-\chi) +L_0\circ E_\chi \right)^{-1}Q_s(L_0^w)^{-1}}_{H^k\rightarrow H^k}<1
\end{align*}
holds for all $s\in\Gamma_{-n,\delta}$. We have
\begin{align*}
    &\norm{L_0^w\left(L_s^w\circ (1-\chi) +L_0\circ E_\chi \right)^{-1}Q_s(L_0^w)^{-1}}_{H^k\rightarrow H^k}\\
    &\le \norm{L_0^w\left(L_s^w\circ (1-\chi) +L_0\circ E_\chi \right)^{-1}}_{H^k\rightarrow H^k}\norm{Q_s}_{D^{k+1}\rightarrow H^k}\norm{(L_0^w)^{-1}}_{H^k\rightarrow D^{k+1}}\\
    &\le \norm{L_0^w\left(L_s^w\circ (1-\chi) +L_0\circ E_\chi \right)^{-1}}_{H^k\rightarrow H^k}\norm{Q_s}_{D^{k+1}\rightarrow H^k}.
\end{align*}
To finish the argument, we just need to show that $Q_s$ can be made arbitrarily small as $\Lambda\rightarrow 0$.
\begin{lemma}\label{extremallimitblackhole}
Fix $s\in\mathbb{C}$. Given $\epsilon>0$, there exists $\Lambda_0>0$ and depending only on $s$ and $k$ such that for all $0<\Lambda<\Lambda_0$,
\begin{align*}
    \norm{Q_s}_{D^{k+1}\rightarrow H^k}<\epsilon.
\end{align*}
\end{lemma}
\begin{proof}
We first fix $k\in\mathbb{N}$ and $0<\epsilon<1$. We also assume $\Lambda$ is sufficiently small so that $1/2<3/(\Lambda r_c^2)<3/2$ and to give the results derived above ($Q_s$ is a holomorphic family of Fredholm operators and $\rho_e<1/3$). For $u\in D^{k+1}(L_s^w)$, we can write
\begin{align*}
    Q_s u=&\left((1-\rho^2)-\frac{3f_{\Lambda}(r_c\rho)}{\Lambda r_c^2}\right)\partial_\rho^2(\chi u)+\left(1-\frac{3}{\Lambda r_c^2}\right)\left(\Delta -\partial_\rho^2\right) (\chi u)\\
    &-\frac{3}{\Lambda r_c^2}\left(r_cw_{\Lambda}'(r_c\rho)+\frac{2}{\rho}w_{\Lambda}(r_c\rho)+\frac{3\alpha_{\Lambda}}{r_c}\right)\partial_\rho (\chi u)+2s\left(\sqrt{\frac{3}{\Lambda r_c^2}(1-f_{\Lambda}(r_c\rho))}-\rho\right)\partial_\rho(\chi u)\\&+s\left(\frac{2}{\rho}\sqrt{\frac{3}{\Lambda r_c^2}\left(1-f_{\Lambda}(r_c\rho)\right)}-\frac{1}{2}\sqrt{\frac{3}{\Lambda r_c^2}}\frac{r_cf_{\Lambda}'(r_c\rho)}{\sqrt{1-f_{\Lambda}(r_c\rho)}}-3\right)\chi u+V^w(r_c\rho)\chi u.
\end{align*}
We deal with each term separately. We have already observed that $V^w(r_c\rho)\rightarrow 0$ in $C^l(\supp\chi)$ as $\Lambda\rightarrow 0$ for each $l\in\mathbb{N}$, so the $V^w\chi u$ term can be made small using that. Since $\Lambda$ is small enough, we have
\begin{align*}
    \frac{3}{\Lambda r_c^2}&\norm{(r_cw_{\Lambda}'(r_c\rho)+2w_{\Lambda}(r_c\rho)/\rho+3\alpha_{\Lambda}/ r_c)\partial_\rho (\chi u)}_{H^k}\\&\le \frac{3}{2}\left(\norm{r_cw_{\Lambda}'(r_c\rho)+2w_{\Lambda}(r_c\rho)/\rho}_{C^k}+\frac{3\alpha_{\Lambda}}{r_c}\right)\norm{u}_{H^{k+1}},
\end{align*}
where $H^l=H^l(U_{\rho_e})$ and $C^l=C^l(\overline{U_{1/3}})$ for each $l\in\mathbb{N}$. From the properties of $w_{\Lambda}$, we see that there exists $\Lambda_0>0$ such that for $\Lambda<\Lambda_0$,
\begin{align*}
    \frac{3}{\Lambda r_c^2}\norm{(r_cw_{\Lambda}'(r_c\rho)+2w_{\Lambda}(r_c\rho)/\rho+3\alpha_{\Lambda}/r_c)\partial_\rho (\chi u)}_{H^k}<\epsilon\norm{u}_{D^{k+1}},
\end{align*}
where $D^{k+1}=D^{k+1}(L_s^w)$. By taking $\Lambda_0$ smaller if necessary, we can also ensure that
\begin{align*}
    \frac{3}{\Lambda r_c^2} (1-f_{\Lambda}(r_c\rho))=\rho^2-\frac{3}{\Lambda r_c^2}w(r_c\rho)+\frac{\alpha_{\Lambda}}{r_c}\rho\in \left[\frac{1}{16},\frac{151}{144}\right]:=I,
\end{align*}
since $3f_{\Lambda}(r_c\rho)/(\Lambda r_c^2)\rightarrow 1-\rho^2$ in $C^l$ for each $l\in\mathbb{N}$. We note that $g(x)=\sqrt{x}$ belongs to $C^{\infty}(I)$ and thus $g, \partial_x g, \dots \partial_x^k g$ are Lipschitz on $I$. Hence we have
\begin{align*}
    \sqrt{\frac{3}{\Lambda r_c^2}(1-f_{\Lambda}(r_c\rho))}\rightarrow \rho \  \text{in}\  C^k\  \text{as}\  \Lambda\rightarrow 0.
\end{align*}
Thus (taking $\Lambda_0$ smaller if necessary) we have
\begin{align*}
    \norm{2s\left(\sqrt{\frac{3}{\Lambda r_c^2}(1-f_{\Lambda}(r_c\rho))}-\rho\right)\partial_\rho(\chi u)}_{H^k}&<\epsilon \norm{u}_{D^{k+1}},\\
    \norm{s\left(\frac{2}{\rho}\sqrt{\frac{3}{\Lambda r_c^2}\left(1-f_{\Lambda}(r_c\rho)\right)}-\frac{1}{2}\sqrt{\frac{3}{\Lambda r_c^2}}\frac{r_cf_{\Lambda}'(r_c\rho)}{\sqrt{1-f_{\Lambda}(r_c\rho)}}-3\right)\chi u}_{H^k}&<\epsilon \norm{u}_{D^{k+1}}.
\end{align*}
We shall now focus on the second order terms. We have already proved that
\begin{align*}
    \norm{(\Delta-\partial_\rho^2)(\chi u)}_{H^k}&\le C\norm{u}_{D^{k+1}},
\end{align*}
where $C$ depends only on $k$ from the discussion in the previous lemma. Thus for $\Lambda<\Lambda_0$
\begin{align*}
    \norm{\left(1-\frac{3}{\Lambda r_c^2}\right)\left(\Delta -\partial_\rho^2\right) (\chi u)}_{H^k}<\epsilon \norm{u}_{D^{k+1}}.
\end{align*}
The radial second order term is
\begin{align*}
    \left((1-\rho^2)-\frac{3f_{\Lambda}(r_c\rho)}{\Lambda r_c^2}\right)\partial_\rho^2(\chi u)=\left(\frac{\Lambda r_c^2}{3}\frac{1-\rho^2}{f_{\Lambda}}-1\right)\frac{3f_{\Lambda}(r_c\rho)}{\Lambda r_c^2}\partial_\rho^2(\chi u).
\end{align*}
For $\rho\in[1/3,1]$, we see that both $f_{\Lambda}(r_c\rho)$ and $1-\rho^2$ vanish only at $\rho=1$ with simple zeroes there, so the fraction in the above expression extends to a smooth function on this interval. We also know that since it obeys that condition, it will converge uniformly to 1 in that interval if and only if its reciprocal does. We write
\begin{align*}
    f_{\Lambda}(r)&=1+w_{\Lambda}(r)+\frac{\Lambda}{3}\alpha_{\Lambda} r-\frac{\Lambda}{3}r^2\\
    &=w_{\Lambda}(r)-w_{\Lambda}(r_c)+\frac{\Lambda}{3}\alpha_{\Lambda}(r-r_c)-\frac{\Lambda}{3}(r^2-r_c^2)\\
    &=\int_{r_c}^r w_{\Lambda}'(x)dx+\frac{\Lambda}{3}\alpha_{\Lambda}(r-r_c)-\frac{\Lambda}{3}(r^2-r_c^2),
\end{align*}
using the fact that $f_{\Lambda}$ vanishes at $r=r_c$. Writing $r=r_c\rho$ and focussing on $\rho\in[1/3,1]$, we see that we have
\begin{align*}
    f_\Lambda(r_c\rho)=(\rho-1)\left(\int_0^1 r_cw_{\Lambda}'(r_c\rho t+r_c(1-t))dt+\frac{\Lambda}{3}\alpha_{\Lambda} r_c-\frac{\Lambda}{3}r_c^2(1+\rho)\right).
\end{align*}
From this we see that
\begin{align*}
    \frac{3}{\Lambda r_c^2}\frac{f_{\Lambda}(r_c\rho)}{1-\rho^2}=1-\frac{1}{1+\rho}\left(\frac{\alpha_{\Lambda}}{r_c}+\int_0^1 r_cw_{\Lambda}'(r_c\rho t+r_c(1-t))dt\right)
\end{align*}
and hence that
\begin{align*}
    \norm{\frac{3}{\Lambda r_c^2}\frac{f_{\Lambda}(r_c\rho)}{1-\rho^2}-1}_{C^k}\le C\left(\frac{\alpha_{\Lambda}}{r_c}+\norm{\int_0^1 r_cw_{\Lambda}'(r_c\rho t+r_c(1-t))dt}_{C^k}\right),
\end{align*}
where $C$ is a constant depending on only on $k$. The $\alpha_{\Lambda}/r_c$ term can be made arbitrarily small by taking $\Lambda$ sufficiently small, so we focus on the other term
\begin{align*}
    \partial_\rho^j\left(\int_0^1 r_cw_{\Lambda}'(r_c\rho t+r_c(1-t))dt\right)=\int_0^1 r_c^{j+1}t^jw_{\Lambda}^{(j+1)}(r_c\rho t+r_c(1-t))dt.
\end{align*}
We know that
\begin{align*}
    |w_{\Lambda}^{(j+1)}(r)|&= |w_{\Lambda}^{(j+1)}(r)-w_0^{(j+1)}(r)+w_0^{(j+1)}(r)|\\
    &\le \frac{\Lambda}{3}\frac{\beta_{j+1}}{r^{j+1}}+\frac{\epsilon}{r^{j+1}}
\end{align*}
for any $\epsilon>0$ provided $r=r_c\rho$ is sufficiently large. Hence we have
\begin{align*}
    \left|\int_0^1 r_c^{j+1}t^jw_{\Lambda}^{(j+1)}(r_c\rho t+r_c(1-t))dt\right|&\le \left(\frac{\Lambda}{3}\beta_{j+1}+\epsilon\right)r_c^{j+1} \left|\int_0^1 \frac{t^j}{(r_c(\rho-1)t+r_c)^{j+1}}dt\right|\\
    &\le \left(\frac{\Lambda}{3}\beta_{j+1}+\epsilon\right)\frac{1}{(j+1)\rho^{j+1}}.
\end{align*}
The right hand side can be made arbitrarily small by taking $\Lambda$ sufficiently small and hence
\begin{align*}
    &\frac{3}{\Lambda r_c^2}\frac{f_{\Lambda}(r_c\rho)}{1-\rho^2}\rightarrow 1 \quad \text{in}\quad C^k([1/3,1])\quad \text{as}\quad \Lambda\rightarrow 0\\
    \Rightarrow &\frac{\Lambda r_c^2}{3}\frac{1-\rho^2}{f_{\Lambda}(r_c\rho)}\rightarrow 1 \quad \text{in}\quad C^k([1/3,1])\quad \text{as}\quad \Lambda\rightarrow 0.
\end{align*}
To reduce clutter, we shall set
\begin{align*}
    P_{\Lambda}(\rho)=1-\frac{\Lambda r_c^2}{3}\frac{1-\rho^2}{f_{\Lambda}(r_c\rho)}
\end{align*}
and observe that
\begin{align*}
    &\left((1-\rho^2)-\frac{3f_{\Lambda}(r_c\rho)}{\Lambda r_c^2}\right)\partial_\rho^2(\chi u)\\&=P_{\Lambda}(\rho)\left(L_0^w(\chi u)+\frac{3}{\Lambda r_c^2}(\Delta-\partial_\rho^2)(\chi u)-\frac{3}{\Lambda r_c^2\rho^2}\partial_\rho(\rho^2(1-f_{\Lambda}))\partial_{\rho}(\chi u)-(V^w+2)\chi u\right),
\end{align*}
which yields the estimate
\begin{align*}
    \norm{\left((1-\rho^2)-\frac{3f}{\Lambda r_c^2}\right)\partial_\rho^2(\chi u)}_{H^k}\le C\norm{P_{\Lambda}}_{C^k} \norm{u}_{D^{k+1}}.
\end{align*}
Note that similarly to before, this $C$ depends only on $k$ provided we already assume $\Lambda$ is small enough. Since $P_{\Lambda}\rightarrow 0$ in $C^k$, there exists $\Lambda_0>0$ depending only on $k$ such that for $\Lambda< \Lambda_0$,
\begin{align*}
    \norm{\left((1-\rho^2)-\frac{3f}{\Lambda r_c^2}\right)\partial_\rho^2(\chi u)}_{H^k}&< \epsilon\norm{u}_{D^{k+1}}.
\end{align*}
Combining all these estimates, we have
\begin{align*}
    \norm{Q_su}_{D^{k+1}\rightarrow H^k}<\epsilon.
\end{align*}
\end{proof}
We know that $(L^w_s\circ (1-\chi)+L_s\circ E_\chi)^{-1}$ has a pole of finite order at $-n$ for $n=1,2,\dots k$ which is enough to establish the result, however we can do better and show that these poles are in fact simple.
\begin{prop}
Fix $k\in\mathbb{N}$. Then for each $n\in \{1,2,\dots k\}$, there exists $\Lambda_0>0$, $A^w_{-1}: H^k\rightarrow D^{k+1}$ a finite rank operator and $A^w_0(s)$ is a holomorphic family of Fredholm operators $H^k\rightarrow D^{k+1}$ such that for $\Lambda<\Lambda_0$,
\begin{align*}
    (L^w_s\circ (1-\chi)+L_s\circ E_\chi)^{-1}=\frac{A^w_{-1}}{s+n}+A^w_0(s)
\end{align*}
in a suitable neighbourhood of $-n$.
\end{prop}
\begin{proof}
First let us note that since the quasinormal modes of $A(s):=L^w_s\circ (1-\chi)+L_s\circ E_\chi$ are smooth, for a quasinormal mode $u$ in the region $2/3<\rho\le 1$
\begin{align*}
A(s)u=0\Rightarrow L_su=0.
\end{align*}
Thus for $2/3<\rho\le 1$, the solutions match the modes we calculated for $L_s$ previously. It also means that the dimension of the kernel of $A(-n)$ is $n^2$ and that the $n^2$ corresponding co-modes $\{\theta_i\}_{i=0}^{n^2}$ are the same distributions from earlier. If we set $\{u_i\}_{i=1}^{n^2}$ to be a basis for the kernel of $A(-n)$ and $\{v_i\}_{i=1}^{n^2}$ a corresponding basis for the kernel of $L_{-n}$ i.e. such that $u_i|_{U_{2/3}}=v_i$, we note that
\begin{align*}
    \theta_i(A'(-n)u_j)=\theta_i(A'(-n)v_j)=\theta_i(L'_{-n}v_j)
\end{align*}
since the $\theta_i$ are concentrated on the horizon where $A(s)=L_s$. We have already seen that $L_s^{-1}$ has simple poles and in a neighbourhood of $-n$ is of the form
\begin{align*}
    L_s^{-1}=\frac{\Pi_1}{s+n}+B(s)
\end{align*}
where $\Pi_1$ is a linear combination of terms of the form $v_i\theta_j$ and $B(s)$ is holomorphic in a neighbourhood of $-n$. For $v\in \ker L_{-n}$ and $s\neq -n$, we see
\begin{align*}
    v&=L_s^{-1}L_sv\\
    &=\frac{\Pi_1L_{-n}v}{s+n}+\Pi_1L'_{-n}v+(s+n)\Pi_1v+(s+n)B(s)(L_{-n}'+s+n)v\\
    &=\Pi_1L'_{-n}v+(s+n)\left(\Pi_1v+B(s)(L_{-n}'+s+n)\right)v
\end{align*}
Taking the limit as $s\rightarrow -n$ gives that for all $v\in\ker L_{-n}$,
\begin{align*}
    v=\Pi_1L'_{-n}v
\end{align*}
Hence it is possible to arrange the $\theta_i$ and $v_i$ such that $\theta_i(L_{-n}'v_j)=\delta_{ij}$.
We relabel the bases of modes $\{u_i\}_{i=1}^{n^2}$ and co-modes $\{\theta_i\}_{i=1}^{n^2}$ of $A(s)$ such that $\theta_i(A'(-n)u_j)=\delta_{ij}$. Now let us suppose the pole is of order $N>1$. In this case, we have
\begin{align*}
    A(s)^{-1}=\sum_{j=1}^N\frac{A^w_{-j}}{(s+n)^j}+A^w_0(s).
\end{align*}
where the $A^w_{-j}$ are of finite rank, $A^w_{-N}\neq 0$ and $A^w_0(s)$ is holomorphic in a neighbourhood of $-n$. Recalling that $A(s)u=(s+n)^2u+(s+n)A'(-n)u+A(-n)u$, we see that
\begin{align*}
    I_{D^{k+1}}&=A(s)^{-1}A(s)\\
    &=\sum_{j=1}^N\frac{A^w_{-j}}{(s+n)^{j-2}}+\sum_{j=1}^N\frac{A^w_{-j}A'(-n)}{(s+n)^{j-1}}+\sum_{j=1}^N\frac{A^w_{-j}A(-n)}{(s+n)^{j}}+A_0^w(s)A(s),\\
    I_{H^k}&=A(s)^{-1}A(s)\\
    &=\sum_{j=1}^N\frac{A^w_{-j}}{(s+n)^{j-2}}+\sum_{j=1}^N\frac{A'(-n)A^w_{-j}}{(s+n)^{j-1}}+\sum_{j=1}^N\frac{A(-n)A^w_{-j}}{(s+n)^{j}}+A(s)A_0^w(s).
\end{align*}
Since the above holds for all $s\neq -n$ in some neighbourhood of $-n$, it follows from equating the terms corresponding to $(s+n)^{-N}$ that the following equations hold:
\begin{align*}
    A^w_{-N}A(-n)=A(-n)A^w_{-N}=0
\end{align*}
This implies that there exist constants $c_{ij}$ such that
\begin{align*}
    A^w_{-N}=\sum_{i,j=1}^{n^2}c_{ij}u_i\theta_j.
\end{align*}
We also have
\begin{align*}
    A^w_{-N}A'(-n)+A^w_{-N+1}A(-n)=0
\end{align*}
so for any $u\in\ker A(-n)$ (in particular $u_m$ for $m=1,2,\dots n^2$),
\begin{align*}
    A^w_{-N}A'(-n)u_m=\sum_{i,j=0}^{n^2}c_{ij}u_i\theta_j(A'(-n)u_m)=\sum_{i=1}^{n^2}c_{im}u_i=0.
\end{align*}
From the linear independence of the $u_i$, the above implies that $c_{ij}=0$ for each $i,j=1,2,\dots n^2$ i.e. $A^w_{-N}=0$, a contradiction. Hence the pole is simple.
\end{proof}
With this in mind, we can proceed similarly to the previous section. For each $n\in\mathbb{N}$, there exists $0<r<1/2$ such that $A^w_0(s)$ is bounded on the disc of radius $r$ around $-n$ and we define the following constants
\begin{align*}
	        C_{k,n}^w=\sup_{s\in \overline{D(-n,r)}}\left\{\norm{L^w_0A_0^w(s)}_{H^k\rightarrow H^k}\right\}
	    \end{align*}
and
\begin{align*}
    \tilde{C}_{k,n}^w=\min \left\{\frac{1}{2\norm{L^w_0A_{-1}}_{H^k\rightarrow H^k}},\frac{1}{C^w_{k,n}}\right\}.
\end{align*}
\begin{prop}\label{mainresultpropblackhole}
For each $n\in \mathbb{N}$ and $0<\delta<r(n,k)$, there exists $\Lambda_0>0$ such that for all $\Lambda<\Lambda_0$, there exists at least one quasinormal frequency $s$ inside $D(-n\sqrt{\Lambda/3}, \delta \sqrt{\Lambda/3})$ i.e.
\begin{align*}
    \left|s+n\sqrt{\frac{\Lambda}{3}} \right|<\delta\sqrt{\frac{\Lambda}{3}}.
\end{align*}
\end{prop}
\begin{proof}
The proof is exactly analogous to the previous results. We see again that on the circle of radius $\delta$ around $-n$ that
\begin{align*}
    \norm{A^w(\s)^{-1}(A^w(\s)-B^w(\s))}_{H^k\rightarrow H^k}&\le \frac{\norm{Q_{\s}}_{D^{k+1}\rightarrow H^k}}{\delta}\norm{L_0A^w_{-1}}_{H^k\rightarrow H^k}+C^w_{k,n}\norm{Q_{\s}}_{D^{k+1}\rightarrow H^k},
\end{align*}
where again $\s=s\sqrt{3/\Lambda}$. We note that by \cref{extremallimitblackhole}, there exists $\Lambda_0>0$ such that for all $\Lambda<\Lambda_0$, $\norm{Q_{\s}}_{D^{k+1}\rightarrow H^k} <\tilde{C}^w_{k,n}\delta$. Thus we have
\begin{align*}
    \norm{A^w(\s)^{-1}(A^w(\s)-B^w(\s))}_{H^k\rightarrow H^k}&< \tilde{C}^w_{k,n}\norm{L_0A^w_{-1}}_{H^k\rightarrow H^k}+\tilde{C}^w_{k,n} C^w_{k,n}\delta\\
    &\le \frac{1}{2}+\delta<1,
\end{align*}
since $\delta<r\le 1/2$. The conditions in \cref{rouche} are met and so the number of quasinormal frequencies (counted with multiplicity) are unchanged. So there exists at least one $s$ such that
\begin{align*}
    \left\lvert s\sqrt{\frac{3}{\Lambda}}+n\right\lvert&<\delta\\
    \Rightarrow \left|s+n\sqrt{\frac{\Lambda}{3}}\right|&<\delta\sqrt{\frac{\Lambda}{3}} <\frac{1}{2}\sqrt{\frac{\Lambda}{3}}.
\end{align*}
\end{proof}
\begin{thm}\label{mainresultthmblackhole}
Let $w_{\Lambda}\in C^{\infty}(0,\infty)$ be a family of functions for $\Lambda\in [0,\Lambda_0]$ for some $\Lambda_0>0$ such that $1+w_0$ has finitely many roots, the largest of which, $r_e$, is simple and such that $w_0<0$ for $r>r_e$. Suppose further that for each $k\in\mathbb{N}_0$ and $\Lambda\in[0,\Lambda_0]$,
\begin{align*}
    (r^k\partial_r^kw_{\Lambda})(r)\rightarrow 0 \ \text{as}\ r\rightarrow \infty
\end{align*}
and let us suppose that there exists $t<r_e$  and $\beta_k$ such that
\begin{align*}
    \sup_{r\ge t}\left|r^k\partial_r^k(w_\Lambda-w_0)(r)\right|<\frac{\Lambda}{3}\beta_k
\end{align*}
for all $k\in\mathbb{N}_0$. Let $\alpha: [0,\Lambda_0]\rightarrow \mathbb{R}$ be a continuous function and write $\alpha(\Lambda)=\alpha_\Lambda$. Let $g$ be the metric defined using the above functions in the same way as at the start of this section and $R$ be its Ricci scalar. Then the equation
\begin{align*}
    -\Box_{g}\psi+\frac{R}{6}\psi=0
\end{align*}
exhibits the phenomenon of zero-damped quasinormal frequencies.
\end{thm}
\begin{proof}
This is simply an application of \cref{mainresultpropblackhole}, noting that for each $\Lambda$ sufficiently small, there is a subset of the quasinormal spectrum $\{s_{n}\}_{n=1}^{\infty}$ such that
\begin{align*}
    \left|s_n+n\sqrt{\frac{\Lambda}{3}}\right|<\frac{1}{2}\sqrt{\frac{\Lambda}{3}}.
\end{align*}
\end{proof}
\subsection{Nearly extremal Reissner-Nordstr\"om-de Sitter}
We expect to be able to prove that zero-damped modes exist for scalar fields minimally coupled to Reissner-Nordstr\"om-de Sitter from the results of \cite{destounis1} (these are called near-extremal modes in this paper), so we shall use this example to present an application of the results obtained above. We consider the Reissner-Nordstr\"om-de Sitter metric
\begin{align*}
    g_{1}=-f_1(r)dt^2+\frac{dr^2}{f_1(r)}+r^2g_{S^2},
\end{align*}
where
\begin{align*}
    f_1(r)=1-\frac{2m}{r}+\frac{q^2}{r^2}-\frac{\Lambda}{3}r^2=-\frac{\Lambda}{3}\frac{(r-r_0)(r-r_-)(r-r_+)(r_c-r)}{r^2}.
\end{align*}
We can write $f_1$ in the second form provided $m, q$ and $\Lambda$ obey suitable conditions so that $f_1$ has four real roots $r_0<0<r_-<r_+<r_c$, so we assume this is the case. This metric is defined on the manifold $\mathcal{M}_1=[0,\infty)\times [r_+,r_c]\times S^2$. By Vieta's formulae, we have
\begin{align}
    r_0+r_-+r_++r_c&=0,\label{zerosum}\\
    -\frac{\Lambda}{3}(r_0r_c+r_+r_-+(r_0+r_c)(r_++r_-))&=1,\\
    -\frac{\Lambda}{3}(r_0r_c(r_++r_-)+r_+r_-(r_c+r_0))&=2m,\\
    -\frac{\Lambda}{3}r_0r_-r_+r_c&=q^2,
\end{align}
Note that these can be inverted, which allows us to write the roots as continuous functions of $\Lambda, m$ and $q$. Thus we can parametrise the spacetime using three of the roots (the fourth is determined by \eqref{zerosum}) and consider the extremal limit using this picture. We define the quantities
\begin{align*}
    \bar{r}&=\frac{r_++r_-}{2}=-\frac{r_0+r_c}{2},\\
    h&=\frac{r_+-r_-}{2},
\end{align*}
and consider the limit where we keep $\bar{r}$ and $r_c$ fixed and take $h\rightarrow 0$. Consider the coordinate transformation inspired by the one used in \cite{Bizon_2013}
\begin{align*}
    \rho=\frac{\gamma r}{r-\bar{r}}, \quad r=\frac{\bar{r}\rho}{\rho-\gamma}
\end{align*}
for some $\gamma (h)>0$ to be determined later. Under this transformation, we see that
\begin{align*}
    r^2=\frac{\bar{r}^2\rho^2}{(\rho-\gamma)^2},
\end{align*}
which motivates the introduction of a conformal factor
\begin{align*}
    \Omega=\frac{\bar{r}}{\rho-\gamma}
\end{align*}
such that
\begin{align*}
    f_1(\rho)&=\frac{\Lambda}{3}\frac{\left((\bar{r}-r_0)\rho+\gamma r_0\right)\left(h\rho+\gamma r_-\right)\left(\gamma r_+-h\rho\right)\left((r_c-\bar{r})\rho-\gamma r_c\right)}{\bar{r}^2\rho^2(\rho-\bar{r})^2}\\
    &=\Omega^2 f_2(\rho).
\end{align*}
A quick calculation yields further that
\begin{align*}
    \frac{dr^2}{f_1(r)}=\frac{\Omega^2d\rho^2}{f_2(\rho)}
\end{align*}
and hence
\begin{align*}
    g_1=\Omega^2\left(-f_2dt^2+\frac{dr^2}{f_2}+\rho^2g_{S^2}\right)=\Omega^2g_2.
\end{align*}
The manifold is mapped to $\mathcal{M}_2=[0,\infty)\times [\gamma r_c/(r_c-\bar{r}),\gamma r_+/h]\times S^2$. Note that the event and cosmological horizons have swapped (see \cref{swappinghorizons}).
\begin{figure}
    \centering
    \begin{subfigure}[b]{0.5\textwidth}
    \centering
    \begin{tikzpicture}
\node (A) at ( -3.5,0){};
\node (B) at (3.5,0){};
\node (C) at (0,3.5){};
\draw (A.center)--  (B.center);
\draw [red](A.center)--node[midway, above, sloped] {$r=r_e$}(C.center);
\draw [blue](B.center)--node[midway,above, sloped] {$r=r_c$} (C.center);
\end{tikzpicture}
\caption{The spacetime $(\mathcal{M}_1,g_1)$}
    \end{subfigure}%
    \begin{subfigure}[b]{0.5\textwidth}
    \centering
    \begin{tikzpicture}
\node (A) at ( -3.5,0){};
\node (B) at (3.5,0){};
\node (C) at (0,3.5){};
\draw (A.center)--  (B.center);
\draw [blue](A.center)--node[midway, above, sloped] {$\rho=\frac{\gamma r_c}{r_c-\bar{r}}$}(C.center);
\draw [red](B.center)--node[midway,above, sloped] {$\rho=\frac{\gamma r_+}{h}$} (C.center);
\end{tikzpicture}
\caption{The spacetime $(\mathcal{M}_2,g_2)$}
    \end{subfigure}
    \caption{The original Reissner-Nordstr\"om-de Sitter spacetime is conformal to another spacetime where the horizons are swapped.}
    \label{swappinghorizons}
\end{figure}
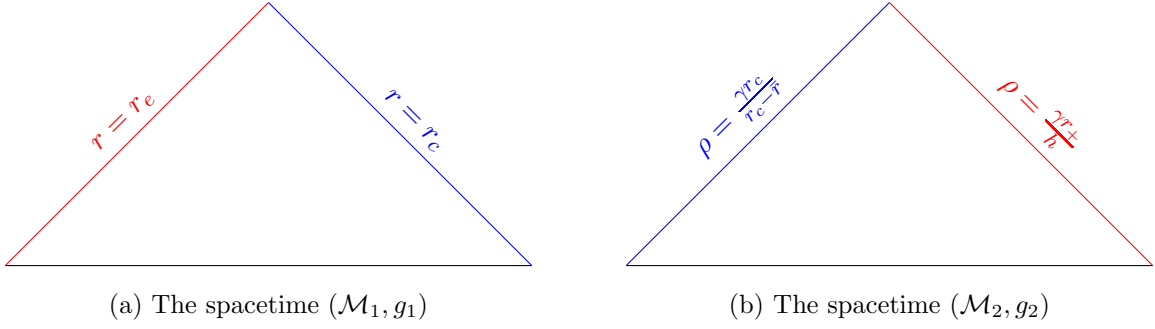
As before, we study the conformal Klein-Gordon equation on $(\mathcal{M}_1,g_1)$:
\begin{align*}
    \left(\Box_{g_1}-\frac{R_1}{6}\right)\psi=0.
\end{align*}
We use a standard result to see that for a smooth function $\psi$, we have
\begin{align*}
    \left(\Box_{g_1}-\frac{R_1}{6}\right)\psi=\Omega^{-3} \left(\Box_{g_2}-\frac{R_2}{6}\right)(\Omega\psi)
\end{align*}
and hence it suffices to consider the quasinormal spectrum of the equation
\begin{align*}
    \left(\Box_{g_2}-\frac{R_2}{6}\right)\varphi=0
\end{align*}
where $\varphi=\Omega\psi$.
\begin{thm}\label{thmrnds}
Consider the conformal Klein-Gordon equation on $(\mathcal{M}_1,g_1)$ as defined above with $m,q, \Lambda$ chosen appropriately. This equation exhibits the phenomenon of zero-damped quasinormal frequencies in the limit as the event horizon becomes extremal.
\end{thm}
\begin{proof}
This is simply an application of \cref{mainresultthmblackhole} to this particular example. Bearing in mind the discussion above, it suffices to check that $f_2$ satisfies the conditions required to apply the theorem.
\begin{align*}
    f_2(\rho)=-\frac{\Lambda}{3\bar{r}^4}\frac{h^2(r_c-\bar{r})(\bar{r}-r_0)}{\rho^2}\left(\rho+\frac{\gamma r_0}{\bar{r}-r_0}\right)\left(\rho-\frac{\gamma r_c}{r_c-\bar{r}}\right)\left(\rho+\frac{\gamma r_-}{h}\right)\left(\rho-\frac{\gamma r_+}{h}\right).
\end{align*}
We define a new ``cosmological constant" $\lambda$ by
\begin{align*}
    \lambda=\frac{\Lambda}{\bar{r}^4}(r_c-\bar{r})(\bar{r}-r_0)h^2
\end{align*}
and we find that $f_2$ takes the form:
\begin{align*}
    f_2(\rho)=-\frac{\lambda}{3}\rho^2+F_1\rho+F_0+\frac{F_{-1}}{\rho}+\frac{F_{-2}}{\rho^2}.
\end{align*}
We again use Vieta's formulae and the relations between the roots, $\bar{r}$ and $h$ to see
\begin{align*}
    F_1&=\frac{4\lambda}{3}\gamma\left(r_c^2+2r_c\bar{r}-2\bar{r}^2\right),\\
    F_0&=\gamma^2\left(\frac{\Lambda}{3} \frac{(r_c-\bar{r})(r_c+3\bar{r})}{\bar{r}^2}+\frac{\lambda}{3}\frac{7\bar{r}^2-6\bar{r}r_c-6r_c^2}{(r_c-\bar{r})(r_c+3\bar{r})}\right)=\gamma^2(a_0+\lambda b_0),\\
    F_{-1}&=2\gamma^3\left(\frac{\Lambda}{3}\frac{(\bar{r}^2-2\bar{r}r_c-r_c^2)}{\bar{r}^2}-\frac{\lambda}{3}\frac{\bar{r}^2-4\bar{r}r_c-2r_c^2}{(r_c-\bar{r})(r_c+3\bar{r})}\right)=2\gamma^3(a_1+\lambda b_1),\\
    F_{-2}&=\gamma^4\left(\frac{\Lambda}{3}\frac{r_c(r_c+2\bar{r})}{\bar{r}^2}-\frac{\lambda}{3}\frac{r_c(r_c+2\bar{r})}{(r_c-\bar{r})(\bar{r}-r_0)}\right)=\gamma^4(a_2+\lambda b_2).
\end{align*}
At this stage we are in a position to make a selection for $\gamma$
\begin{align*}
    \gamma=\frac{1}{\sqrt{a_0+\lambda b_0}}
\end{align*}
for $\lambda\in [0,\lambda_0]$ where $\lambda_0<-b_0/a_0$ is taken sufficiently small. We define
\begin{align*}
    \alpha_{\lambda}=\frac{4(r_c^2+2r_c\bar{r}-2\bar{r}^2)}{\sqrt{a_0+\lambda b_0}}
\end{align*}
and note that it is a continuous function of $\lambda$ and converges to a well-defined, finite limit as $\lambda\rightarrow 0$. Thus we have
\begin{align*}
    f_2(\rho)&=1+\frac{\lambda}{3}\alpha_{\lambda}\rho-\frac{\lambda}{3}\rho^2+\frac{2a_1+2\lambda b_1}{(a_0+\lambda b_0)^{3/2}}\frac{1}{\rho}+\frac{a_2+\lambda b_2}{(a_0+\lambda b_0)^2}\frac{1}{\rho^2}\\
    &=1+\frac{\lambda}{3}\alpha_{\lambda}\rho-\frac{\lambda}{3}\rho^2+w_{\lambda}(\rho).
\end{align*}
Now we check that $w_{\lambda}$ satisfies the conditions defined above. Some of these properties are clear: since $w_\lambda$ is a polynomial in $1/\rho$, naturally $w_{\Lambda}\in C^{\infty}(0,\infty)$ and there exists $\lambda_0>0$ such that for each $k\in\mathbb{N}_0$ and $\lambda\in [0,\lambda_0]$,
\begin{align*}
    \rho^k(\partial_\rho^kw_{\lambda})(\rho)\rightarrow 0 \quad \text{as} \quad \rho\rightarrow \infty.
\end{align*}
Furthermore it is clear that $1+w_0$ has finitely many zeroes in $(0,\infty)$ (0, 1 or 2), but we need to check that it has at least one. Since $r_c>\Bar{r}$, we see that $a_1<0$ and $a_2>0$ and hence any real roots of
\begin{align*}
    1+w_0(\rho)&=1+\frac{2a_1/a_0^{3/2}}{\rho}+\frac{a_2/a_0^2}{\rho^2}\\
    &=\frac{\rho^2+2a_1/a_0^{3/2}\rho+a_2/a_0^2}{\rho^2}\\
    &=\frac{\rho^2-2\mu\rho+\nu}{\rho^2}
\end{align*}
will be positive. So it suffices to consider the discriminant of the quadratic,
\begin{align*}
    \text{Disc}=\frac{4}{a_0^3}(a_1^2-a_0a_2).
\end{align*}
Since we are only interested in the sign of this quantity, we can multiply out positive factors (such as $\Lambda/(9\bar{r}^4)$) and consider
\begin{align*}
    \text{Disc}'&=(\bar{r}^2-2\bar{r}r_c-r_c^2)^2-r_c(r_c+2\bar{r})(r_c-\bar{r})(r_c+3\bar{r})\\
    &=\bar{r}^4+\bar{r}^2r_c^2+2\bar{r}^3r_c\\
    &=\bar{r}^2(\bar{r}+r_c)^2>0.
\end{align*}
Hence $1+w_0$ has two positive roots in $(0,\infty)$ and furthermore these are simple. The largest root is $\rho_e=\mu+\sqrt{\mu^2-\nu}$. Plugging this in to our expression for $w_0$, we see
\begin{align*}
    w_0(\rho)&=-\frac{1}{\rho^2}(2\mu\rho-\nu)\\
    &=-\frac{1}{\rho^2}\left(2\mu^2+2\mu\sqrt{\mu^2-\nu}-\nu+2\mu(\rho-\rho_e)\right)\\
    &=-\frac{1}{\rho^2}\left(\rho_e^2+2\mu(\rho-\rho_e)\right),
\end{align*}
which is negative for $\rho>\rho_e$. Now we check the final condition:
\begin{align*}
   \left|\rho^k\partial_\rho^k (w_\lambda-w_0)(\rho)\right|\le\frac{k!\left|F_{-1}(\lambda)-F_{-1}(0)\right|}{\rho^{k+1}}+\frac{(k+1)!\left|F_{-2}(\lambda)-F_{-2}(0)\right|}{\rho^{k+2}}.
\end{align*}
In a sufficiently small neighbourhood of $\lambda=0$, the $F_i$ are continuously differentiable with bounded derivatives and hence Lipschitz. Thus, there exists a constant $C>0$ such that
\begin{align*}
    |F_j(\lambda)-F_j(0)|\le \frac{\lambda}{3}C
\end{align*}
for $j=-1,-2$. After fixing some $0<t<\rho_e$ defined above, we see that
\begin{align*}
   \sup_{\rho\ge t}\left|\rho^k\partial_\rho^k (w_\lambda-w_0)(\rho)\right|&\le\frac{k!\left|F_{-1}(\lambda)-F_{-1}(0)\right|}{t^{k+1}}+\frac{(k+1)!\left|F_{-2}(\lambda)-F_{-2}(0)\right|}{t^{k+2}}\\
   &\le \frac{\lambda}{3}\beta_k
\end{align*}
for some $\beta_k>0$.
\end{proof}
\section{Acknowledgements}
The author would like to thank Claude Warnick for motivating the contents of this paper, for countless useful suggestions throughout its preparation and for reading through several drafts. He would also like to thank Maciej Dunajski and Jorge Santos for their helpful comments, and Owain Salter Fitz-Gibbon and Fred Alford for many productive discussions.\\\\
This work was funded by STFC DTP (2267798)
\printbibliography
\end{document}